\documentclass[preprint,12pt,margin=1in]{elsarticle}

\usepackage{amsmath,amssymb}
\usepackage{graphicx}
\usepackage[colorlinks=true,allcolors=blue]{hyperref}
\usepackage{orcidlink}

\usepackage{graphicx}
\usepackage{amsthm}   % put this in the preamble, after amsmath
\usepackage{setspace}
\usepackage{natbib}
\usepackage{hyperref}
\usepackage{algorithm}
\usepackage{algorithmicx}
\usepackage{algcompatible} % for compatibility with algorithmic syntax
\usepackage[font=scriptsize]{caption}
\usepackage{enumitem}
\usepackage{booktabs}
\usepackage{rotating}    % for sidewaystable  % for \toprule, etc.
\usepackage{longtable}   % if needed for multi-page tables
\usepackage[margin=1in]{geometry}
\usepackage{adjustbox}   % for scaling   
\usepackage{algorithm}
\usepackage[noend]{algpseudocode}
\usepackage{enumitem}           % for your custom enumerate settings
\usepackage{xcolor} 
\usepackage{booktabs}
\usepackage{ulem}

\newtheorem{theorem}{Theorem}
\newtheorem{assumption}{Assumption}
\newtheorem{remark}{Remark}

\newtheorem{proposition}{Proposition}

\begin{document}

\begin{frontmatter}
\title{A novel multi-exposure-to-multi-mediator mediation model for imaging genetic study of brain disorders}

% Authors and ORCID
\author[1,3]{Neng Wang}
\author[2]{Eric V. Slud\,\corref{cor1}\orcidlink{0000-0002-0647-9193}}
\author[3]{Tianzhou Ma\,\corref{cor2}\orcidlink{0000-0003-3605-0811}}

% Corresponding author mark
\cortext[cor1]{Corresponding author. Email: slud@umd.edu}
\cortext[cor2]{Corresponding author. Email: tma0929@umd.edu}

% Affiliations
\address[1]{Department of Mathematics, University of Maryland, College Park, MD 20742, USA}
\address[2]{Department of Mathematics, University of Maryland, College Park, MD 20742, USA}
\address[3]{Department of Epidemiology and Biostatistics, University of Maryland, College Park, MD 20742, USA}

\begin{abstract}
 Common psychiatric and brain disorders are highly heritable and affected by a number of genetic risk factors, yet the mechanism by which these genetic factors contribute to the disorders through alterations in brain structure and function remain poorly understood. Contemporary imaging genetic studies integrate genetic and neuroimaging data to investigate how genetic variation contributes to brain disorders via intermediate neuroimaging endophenotypes. However, the large number of potential exposures (genes) and mediators (neuroimaging features) pose new challenges to the traditional mediation analysis. In this paper, we propose a novel multi-exposure-to-multi-mediator mediation model that integrates genetic, neuroimaging and phenotypic data to investigate the ``gene-neuroimaging-brain disorder'' mediation pathway. Our method jointly reduces the dimensions of exposures and mediators into low-dimensional aggregators where the mediation effect is maximized. We further introduce sparsity into the loadings to improve the interpretability. To target the bi-convex optimization problem, we implement an efficient alternating direction method of multipliers algorithm with block coordinate updates. We provide theoretical guarantees for the convergence of our algorithm and establish the asymptotic properties of the resulting estimators. Through extensive simulations, we demonstrate that our method outperforms other competing methods in recovering true loadings and true mediation proportions across a wide range of signal strengths, noise levels, and correlation structures. We further illustrate the utility of the method through a mediation analysis that integrates genetic, brain functional connectivity and smoking behavior data from UK Biobank, and identifies critical genes that impact nicotine dependence via changing the functional connectivity in specific brain regions.

\end{abstract}
\end{frontmatter}

\section{Introduction}

Common psychiatric and brain disorders are highly heritable and polygenic, influenced by numerous genetic risk factors \citep{brainstorm2018analysis,kaufmann2019common}. Previous genome-wide association studies (GWAS) and transcriptome-wide association studies (TWAS) have identified critical genetic risk factors at both SNP and gene levels and improved our understanding of the genetic influence on these disorders \citep{brainstorm2018analysis,horwitz2019decade,ye2021meta,hatoum2023multivariate}. Yet, how these genes impact the disorders via altering the brain structure or function are not well understood. Studies using the magnetic resonance imaging (MRI) data have found heritable neuroimaging features including both structural and functional signatures of these brain disorders \citep{elliott2018genome}, laying out the basis of imaging genetic studies of brain disorders that examine how genetic variation influences brain disorders through intermediate neuroimaging endophenotypes \citep{bi2017genome}.

Mediation analysis is a statistical framework used to delineate how exposures affect outcomes through intermediate variables known as mediators, providing key insights into the causal mechanism \citep{mackinnon2000equivalence}. Traditional mediation analysis deals with one single exposure and one single mediator at a time. With the advancement of technology, contemporary studies frequently involve a large number of exposures and mediators, e.g. genetic and neuroimaging features in imaging genetic studies (Fig \ref{fig:gene_mri_disorder}), which present new methodological challenges that classical mediation approaches are ill-equipped to address. A number of methods have been developed to address these challenges in mediation analysis. One major class of methods considers performing screening or regularization procedures to reduce the number of mediators and/or exposures first before the mediation analysis \citep{zhang2016estimating,zhao2022pathway,perera2022hima2,zhang2022high}. While effective for large mediator sets, these methods usually separate the mediator selection from mediation effect estimation, potentially overlooking joint mediator effects and introducing false positives. Another class of methods applies dimension reduction on high-dimensional mediators, including principal component analysis and its sparse variants \citep{huang2016hypothesis,gao2019testing,zhao2020multimodal}. However, the dimension reduction step is independent of mediation analysis thus the mapped low-dimensional space is not directly related to the mediation pathway. \cite{chen2018high} considered seeking linear mediator combinations maximizing mediated likelihood within structural equation models using the Directions of Mediation (DMs) approach. \cite{lee2024new} proposed a multi-mediator method that mapped the high-dimensional mediators to a space that maximizes the mediation proportion, and \cite{dai2024integrative} introduced the Mediation Multiple Pathways (MMP) framework, which aggregates evidence across multiple linear mediation directions by enforcing orthogonality and controlling false discovery in both marginal and conditional indirect effects. Few methods have effectively balanced between model flexibility, model interpretability and computational efficiency in these new mediation analysis problems.

 In imaging genetic study of brain disorders, both genetic risk factors and neuroimaging features are potentially high-dimensional. Genes are correlated and working together to carry similar functions \citep{sullivan2018psychiatric}. At the same time, the neuroimaging features (e.g. MRI features from different brain regions) are inherently spatially-correlated \citep{smith2015positive} (Fig \ref{fig:gene_mri_disorder}). We need to consider the simultaneous selection of exposures and mediators from a large number of correlated candidates. In addition, methods creating latent mediator composites often sacrifice interpretability, complicating the translation of findings into actionable biological insights. Lastly, heavily penalized models may exhibit instability, with sparse solutions highly sensitive to data perturbations and tuning parameter choices, thus limiting reproducibility \citep{huang2016hypothesis,zhao2020multimodal}.

To overcome these challenges, in this paper, we introduced a novel multi-exposure-to-multi-mediator mediation model that integrates genetic, neuroimaging and phenotypic data to investigate the ``gene-neuroimaging-brain disorder'' pathway (Figure~\ref{fig:gene_mri_disorder}). In our model, we jointly reduced the dimensions of exposures and mediators into a lower-dimensional space that captures the strongest mediating effect of brain imaging in the mediation pathway. In addition, we imposed structured sparsity constraints on loadings to select a small subset of exposures and mediators, maintaining interpretability and stability. To efficiently estimate the proposed model, we developed a tailored optimization algorithm based on alternating direction method of multipliers (ADMM). We also theoretically showed the convergence of our algorithm and asymptotic properties of our estimates. Our method consistently outperformed competing methods in simulation studies, showing improved power, enhanced mediator selection stability, and better interpretability. Additionally, we illustrated the utility of our proposed method through an integrative mediation analysis of gene expression, brain functional connectivity, and smoking behavior data, demonstrating how critical genes (e.g. \textit{CHRNA5-CHRNA3-CHRNB4} gene cluster) impact nicotine dependence via changing the brain functional connectivity between middle frontal gyrus, precentral gyrus, thalamus and superior parietal lobule regions.

The remainder of this paper is structured as follows. Section 2 details our methodological framework, estimation procedure and theoretical guarantee. Section 3 presents extensive simulations comparing our method against existing approaches. Section 4 applies our method to a real data example investigating the role of brain functional connectivity features that mediate the genetic effect on nicotine addiction in UK Biobank. Section 5 concludes with a discussion of the findings, potential limitations, and avenues for future research.

\begin{figure}[htbp]
\centering
\includegraphics[scale=0.5]{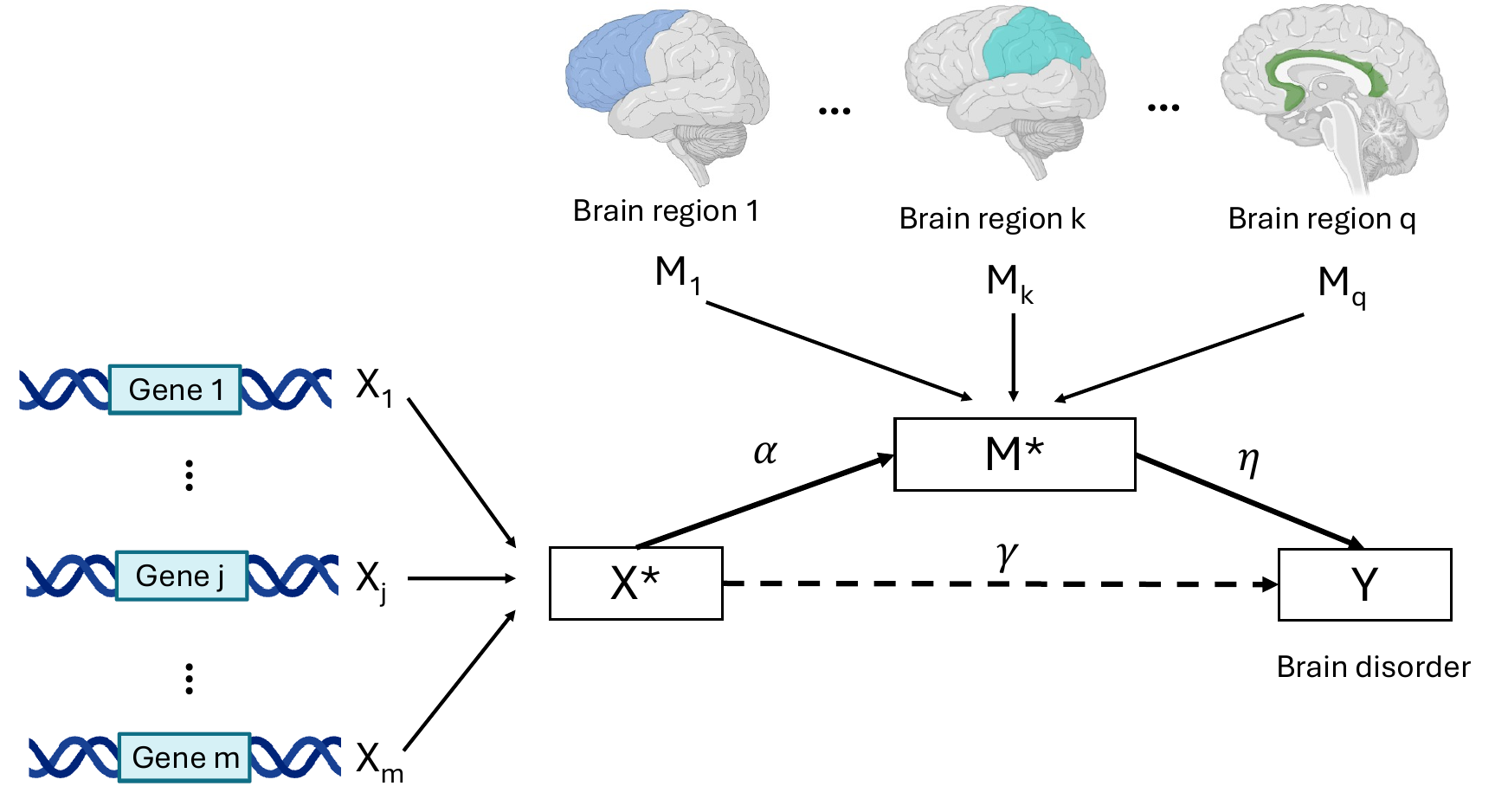}
\caption{The proposed multi-exposure-to-multi-mediator framework to investigate how multiple genetic risk factors (e.g. gene expression from TWAS) impact brain disorders via changing multiple neuroimaging features (e.g. MRI features from different brain regions). In this model, genes act as exposures, neuroimaging features act as mediators, and brain disorder is the outcome.}
\label{fig:gene_mri_disorder}
\end{figure}

\section{Method} 
\subsection{A novel multi-exposure-to-multi-mediator model} \label{model}
{We consider a multi-exposure-to-multi-mediator model to investigate the gene-neuroimaging-brain disorder mediation pathway (Fig \ref{fig:gene_mri_disorder}). We are specifically interested in knowing how a group of functionally related genetic risk factors impact the brain disorder via changing the function and structure of the brain. Denote by $\mathbf{X}=(X_1,\ldots,X_m)$ the $n\times m$ matrix of exposures of $m$ related genetic risk factors of a brain disorder for $n$ subjects. These genetic risk factors could refer to the genotypes of a set of related SNPs (e.g., in the same linkage disequilibrium (LD) region) associated with the outcome from a GWAS, or the predicted expression levels (also known as genetically regulated expression (GReX) \citep{gamazon2015gene}) of a set of genes associated with the outcome from a TWAS. Denote by $\mathbf{M} = (M_1, \ldots, M_q)$ the $n \times q$ matrix of $q$ neuroimaging mediators for $n$ subjects, where each $M_k$ ($1\le k\le q$) may represent an imaging-derived phenotype (IDP) such as functional connectivity or cortical thickness measured across $q$ brain regions from structural or functional MRI. Denote by $\mathbf{Y} = (Y_1, \ldots, Y_n)^\top$ the vector of phenotypic outcome. Here we assume a linear mediation model so $Y$ needs to be a quantitative measure (e.g. cigarettes per day to capture nicotine dependence in real data example). When the outcome is categorical, it is readily extensible to a logistic or multinomial outcome model. In real data application, we adjust for known covariates (e.g., age, sex, and ancestry principal components) by regressing them out from each column of $\mathbf{X}$, $\mathbf{M}$, and $\mathbf{Y}$, i.e. each is replaced by its residuals from a linear model that includes the covariates. While this adjustment is standard practice, it introduces dependencies among the entries of $\mathbf{Y}$ that must be acknowledged; the residualized outcomes are not strictly independent. Nonetheless, this preprocessing step ensures that subsequent mediation analysis targets variation in $\mathbf{Y}$, $\mathbf{X}$, and $\mathbf{M}$ that is not explained by nuisance variables.

Given the large sets of correlated exposures and correlated mediators, we adopt a reduced-rank modeling approach by assuming that the mediation effect is primarily conveyed through a few dominant directions in the exposure and mediator spaces. This assumption is supported by the empirical success of dimension reduction techniques such as principal component analysis (PCA) and canonical correlation analysis (CCA), which often reveal low-dimensional structures capturing most of the variability or association. Specifically, we posit that the mediation effect is concentrated along one principal direction in the exposure space and one in the mediator space. These directions serve as linear combinations that summarize the primary pathway through which the exposures influence the outcome via the mediators \citep{zhao2020multimodal,zhou2020sparse}. Specifically, we define low‐dimensional aggregators:
\begin{equation}
  \mathbf{X}^* = \mathbf X\,\mathbf a,\qquad
  \mathbf{M}^* = \mathbf M\,\mathbf b,
\end{equation}
{where $\mathbf a\in\mathbb R^m$ and $\mathbf b\in\mathbb R^q$ are weight vectors to be estimated. In this model, we consider that only one latent variable (e.g. one-dimensional aggregator) exists for exposures and mediators (i.e. $\mathbf{X}^* = (X^*_1,\ldots, X^*_n)$ and $\mathbf{M}^* = (M^*_1,\ldots, M^*_n)$ for $n$ subjects), respectively. In a real application, there might exist more than one latent variable and we may consider extending our model to a scenario with multiple latent variables. As we will see next, the general idea of capturing the primary mediation pathway still applies. To ensure identifiability, these aggregators are subject to the unit-norm constraints}
%\begin{equation}
%\label{eq:constraints}
$\|\mathbf X\,\mathbf a\|_2=1$ and $\|\mathbf M\,\mathbf b\|_2=1$,
%\end{equation}
which fix scale, removing ambiguity due to rescaling.

We assume a linear mediation structure on the aggregated variables $\mathbf{X}^*$ and $\mathbf{M}^* $, following is a standard formulation in the mediation literature \citep{baron1986moderator, mackinnon2007mediation, imai2010identification}:
\begin{align}
\mathbf{M}^* &= \alpha\,\mathbf{X}^* + \boldsymbol{\varepsilon}_M,
\label{eq:med1}\\
\mathbf{Y}   &= \gamma\,\mathbf{X}^* + \eta\,\mathbf{M}^* + \boldsymbol{\varepsilon}_Y,
\label{eq:med2}
\end{align}
where $\boldsymbol{\varepsilon}_M \sim {\cal N}(0,\sigma_M^2 \, I_{n \times n})$, $\boldsymbol{\varepsilon}_Y \sim {\cal N}(0,\sigma_Y^2 \, I_{n \times n})$, and \(\theta=(\mathbf a,\mathbf b,\alpha,\gamma,\eta,\sigma_M^2,\sigma_Y^2)\) is the set of unknown parameters. The error terms are assumed to be independent with exposures and mediators, i.e. $\varepsilon_M \perp \mathbf X, \quad \varepsilon_Y \perp (\mathbf X, \mathbf M)$ (see section 2.3 for details). Without loss of generality, we assume \(\alpha>0\) and \(\operatorname{sign}(\eta)=\operatorname{sign}(\gamma)\), ensuring the mediation proportion \(\mathrm{MP}=\alpha\eta/\tau\) where \(\tau=\gamma+\alpha\eta\) is well defined and bounded by \(1\).

Given the latent aggregates \(\mathbf X^*=\mathbf X\mathbf a\) and \(\mathbf M^*=\mathbf M\mathbf b\),
the ordinary least squares (OLS) coefficients of $\alpha$, $\gamma$, $\eta$ and $\tau$ are defined as the unique minimizers of $\|\mathbf M^*-\alpha\,\mathbf X^*\|_2^2$, $\|\mathbf Y-\gamma\,\mathbf X^*-\eta\,\mathbf M^*\|_2^2$ and $\|\mathbf{Y} - \tau\mathbf{X}\mathbf{a}\|_2^2$. 
Under the constraints \(\|\mathbf X^*\|_2=\|\mathbf M^*\|_2=1\), the normal equations yield
\[
\hat{\alpha}
= (\mathbf X^*)^{\!\top}\mathbf M^*,
\quad
\hat{\gamma}
= \frac{(\mathbf X^*)^{\!\top}\mathbf Y \;-\; \hat{\alpha}\,(\mathbf M^*)^{\!\top}\mathbf Y}{1-\hat{\alpha}^{2}},
\]
\[
\hat{\eta}
= \frac{(\mathbf M^*)^{\!\top}\mathbf Y \;-\; \hat{\alpha}\,(\mathbf X^*)^{\!\top}\mathbf Y}{1-\hat{\alpha}^{2}},
\quad
\hat{\tau}
%= \arg\min_{\tau}\;\|\mathbf Y-\tau\,\mathbf X^*\|_2^{2}
=
(\mathbf X^*)^{\!\top}\mathbf Y
=
\mathbf a^{\!\top}\mathbf X^{\!\top}\mathbf Y.
\]

\medskip

However, the latent aggregates $\mathbf{X}^* = \mathbf{X}\mathbf{a}$ and $\mathbf{M}^* = \mathbf{M}\mathbf{b}$ depend on unknown weight vectors $\mathbf{a} \in \mathbb{R}^m$ and $\mathbf{b} \in \mathbb{R}^q$. Out of the possible weight vectors, we wish to capture the primary gene-neuroimaging-brain disorder mediation pathway, thus targeting at a large mediated effect (via the term involving \(\hat\alpha\hat\eta/\hat\tau\)) while encouraging sparsity in \(\mathbf a,\mathbf b\) for the best interpretability of the resulting latent aggregates. To balance three competing aims—(i) accurate estimation of the linear mediation model,  
(ii) a large mediation proportion, and (iii) sparsity for interpretability, we maximize a penalized (Lagrangian) sum of squared error objective with both a penalty on mediation proportions and sparsity penalties on weight vectors.  Formally, for observed data \((\mathbf{Y},\mathbf{X},\mathbf{M})\) and tuning parameters \(\lambda_n,\lambda_a,\lambda_b>0\), we define the following loss function:  
\begin{equation}\label{objective}
\footnotesize
\begin{aligned}
\mathcal{L}\bigl(\mathbf{a},\mathbf{b}\,;\,
\mathbf{Y},\mathbf{X},\mathbf{M},
\lambda_n,\lambda_a,\lambda_b\bigr)
 \;=\;&
\frac{1}{2n}\!\left\{
SSR_{\mathbf Y \sim \mathbf X^*}
 \;+\;
SSR_{\mathbf M^* \sim \mathbf X^*}
 \;+\;
\frac{SSR_{\mathbf Y \sim \mathbf X^*,\mathbf M^*}}{1-\hat{\alpha}^{2}}
\right\}
\\[2pt]
&\;-\;
\lambda_n\!\left(\frac{\hat{\alpha}\,\hat{\eta}}{\hat{\tau}}\right)
 \;+\;
\lambda_a\|\mathbf{a}\|_{1}
 \;+\;
\lambda_b\|\mathbf{b}\|_{1},\
\end{aligned}
\end{equation}
{where the first three sum of squared residual (SSR) terms correspond to the estimation error from outcome on latent exposure ($SSR_{\mathbf{Y} \sim \mathbf{X}^*} = \|\mathbf{Y} - \hat{\tau}\,\mathbf{X}\mathbf{a}\|_2^2$), latent mediator on latent exposure ($SSR_{\mathbf{M}^* \sim \mathbf{X}^*} = \|\mathbf{M}\mathbf{b} - \hat{\alpha}\,\mathbf{X}\mathbf{a}\|_2^2$), and outcome on latent exposure and mediator ($SSR_{\mathbf{Y} \sim \mathbf{X}^*, \mathbf{M}^*} = \|\mathbf{Y} - \hat{\gamma}\,\mathbf{X}\mathbf{a} - \hat{\eta}\,\mathbf{M}\mathbf{b}\|_2^2$), respectively. The denominator $1 - \hat{\alpha}^2$ in the third term compensates for the correlation between $\mathbf{X}^*$ and $\mathbf{M}^*$, partially compensating variance inflation due to multicollinearity. The term involving $\hat{\alpha} \hat{\eta} / \hat{\tau}$ aims to maximize the estimated mediation proportion, while the $\ell_1$ penalties encourage sparse selections of manifest exposures and mediators in the latent variables, yielding interpretable models with a small number of active features \citep{tibshirani1996regression}. This objective can, under certain regularity conditions in Section~\ref{subsec:conv_theorem}, recover the aggregation directions up to scaling and sign.

\begin{remark} 

The parameters $(\mathbf{a}, \mathbf{b}, \alpha, \gamma, \eta)$ are identifiable only under the assumptions that $\tau = \alpha + \gamma \eta \ne 0$ and that, conditional on the exposures, no mediator direction orthogonal to $\mathbf{M}\mathbf{b}$ is associated with $\mathbf{X}$ (Assumption~\ref{asmp:full_model} in Section \ref{subsec:conv_theorem}). Further justification and a proof of identifiability of the parameters (Prop.~\ref{prop:identifiability}) can be found in Section \ref{subsec:conv_theorem}.

\end{remark}

\subsection{Optimization using alternating direction method of multipliers (ADMM) algorithm} \label{subsec:opt}

{The objective function in (\ref{objective}) results in a bi-convex optimization problem in $\mathbf{a}$ and $\mathbf{b}$, we propose to solve using an alternating direction method of multipliers (ADMM) algorithm \citep{boyd2011distributed}. Each subproblem is convex and admits a closed-form soft-thresholding update, enabling efficient computation. The method thereby produces interpretable weight vectors that reveal which exposures and mediators contribute most significantly to the mediation pathway. To facilitate efficient updates under $\ell_1$-regularization, we introduce auxiliary variables $\mathbf{z}_a \in \mathbb{R}^m$ and $\mathbf{z}_b \in \mathbb{R}^q$ that decouple the sparsity-inducing penalties from the smooth residual terms. The augmented Lagrangian for this problem introduces scaled dual variables $\mathbf{u}_a$ and $\mathbf{u}_b$, and includes quadratic penalties for constraint violations:}
\begin{align*}
\mathcal{L}_\rho\bigl(\mathbf{a}, \mathbf{b}, \mathbf{z}_a, \mathbf{z}_b, \mathbf{u}_a, \mathbf{u}_b\bigr) =\;&
\frac{1}{2n} \left\| \mathbf{Y} - \hat{\tau}\, \mathbf{X} \mathbf{a} \right\|^2 
+ \frac{1}{2n} \left\| \mathbf{M} \mathbf{b} - \hat{\alpha}\, \mathbf{X} \mathbf{a} \right\|^2 
+ \frac{1}{2n}\,\frac{\left\| \mathbf{Y} - \hat{\gamma}\, \mathbf{X} \mathbf{a} - \hat{\eta}\, \mathbf{M} \mathbf{b} \right\|^2}{1-\hat{\alpha}^{2}} \\
&- \lambda_n \frac{\hat{\alpha} \hat{\eta}}{\hat{\tau}}
+ \lambda_a \|\mathbf{z}_a\|_1 + \lambda_b \|\mathbf{z}_b\|_1 \\
&+ \frac{\rho}{2} \left\|\mathbf{a} - \mathbf{z}_a + \mathbf{u}_a\right\|_2^2 
+ \frac{\rho}{2} \left\|\mathbf{b} - \mathbf{z}_b + \mathbf{u}_b\right\|_2^2,
\end{align*}
where $\rho$ is the penalty parameter in the augmented Lagrangian that controls the convergence and balance between primal and dual residuals in ADMM algorithm. These quadratic terms enforce the consensus constraints $\mathbf{a}=\mathbf{z}_a$ and $\mathbf{b}=\mathbf{z}_b$ in the scaled-dual ADMM form.

ADMM proceeds by iteratively updating (1) the \emph{primal} variables $(\mathbf{a},\mathbf{b})$, (2) the \emph{auxiliary} variables $(\mathbf{z}_a,\mathbf{z}_b)$, and (3) the \emph{dual} variables $(\mathbf{u}_a,\mathbf{u}_b)$. 
The scalar quantities $(\tau,\alpha, \gamma,\eta)$ are \emph{not} updated by ADMM steps; instead, at each iteration, they are \emph{profiled in closed form} from the current $(\mathbf{a},\mathbf{b})$ (see Remark~\ref{rem:profile_admm}). Convergence is assessed via the primal and dual residuals
\[
\mathbf{r}^{(k+1)}=\begin{bmatrix}\mathbf{a}^{(k+1)}-\mathbf{z}_a^{(k+1)}\\[2pt]\mathbf{b}^{(k+1)}-\mathbf{z}_b^{(k+1)}\end{bmatrix},
\qquad
\mathbf{s}^{(k+1)}=\rho\begin{bmatrix}\mathbf{z}_a^{(k+1)}-\mathbf{z}_a^{(k)}\\[2pt]\mathbf{z}_b^{(k+1)}-\mathbf{z}_b^{(k)}\end{bmatrix},
\]
and we terminate when $\|\mathbf{r}^{(k+1)}\|_\infty\le\varepsilon_{\mathrm{pri}}$ and $\|\mathbf{s}^{(k+1)}\|_\infty\le\varepsilon_{\mathrm{dual}}$; in practice we also stop when the relative change in the augmented Lagrangian falls below $10^{-6}$ for 20 consecutive iterations, or after 500 iterations, whichever occurs first. Detailed optimization steps can be found in Algorithm \ref{alg:admm}. 

\begin{algorithm}[htbp]
  \caption{ADMM-based optimization for penalized mediation estimation in our multi-exposure-to-multi-mediator model}
  \label{alg:admm}
  \begin{algorithmic}[1]
    \Statex \textbf{Initialize:} 
      $\mathbf{a}^{(0)},\ \mathbf{b}^{(0)},\ 
       \mathbf{z}_a^{(0)}=\mathbf{a}^{(0)},\ 
       \mathbf{z}_b^{(0)}=\mathbf{b}^{(0)},\ 
       \mathbf{u}_a^{(0)}=\mathbf{0},\ 
       \mathbf{u}_b^{(0)}=\mathbf{0}$
    \For{$k = 0,1,2,\dots$ \textbf{until convergence}}
      \State \textbf{Update $\mathbf{a}$ by blockwise coordinate descent:}
      \[
        \mathbf{a}^{(k+1)}
        = \arg\min_{\mathbf{a}}
          \mathcal{L}_\rho\bigl(\mathbf{a},\mathbf{b}^{(k)},\mathbf{z}_a^{(k)},\mathbf{z}_b^{(k)},\mathbf{u}_a^{(k)},\mathbf{u}_b^{(k)}\bigr)
      \]

      \State \textbf{Update $\mathbf{b}$ by blockwise coordinate descent:}
      \[
        \mathbf{b}^{(k+1)}
        = \arg\min_{\mathbf{b}}
          \mathcal{L}_\rho\bigl(\mathbf{a}^{(k+1)},\mathbf{b},\mathbf{z}_a^{(k)},\mathbf{z}_b^{(k)},\mathbf{u}_a^{(k)},\mathbf{u}_b^{(k)}\bigr)
      \]

      \State \textbf{Update $\mathbf{z}_a$ (soft-thresholding):}
      \[
        \mathbf{z}_a^{(k+1)}
        = \arg\min_{\mathbf{z}_a}
          \lambda_a\|\mathbf{z}_a\|_1
          + \frac{\rho}{2}\bigl\|\mathbf{z}_a - (\mathbf{a}^{(k+1)} + \rho^{-1}\mathbf{u}_a^{(k)})\bigr\|_2^2
      \]
      \[
        \Longrightarrow\quad
        \bigl(z_a^{(k+1)}\bigr)_i
        = S_{\lambda_a/\rho}\bigl((a^{(k+1)})_i + \rho^{-1}(u_a^{(k)})_i\bigr).
      \]

      \State \textbf{Update $\mathbf{z}_b$ (soft-thresholding):}
      \[
        \mathbf{z}_b^{(k+1)}
        = \arg\min_{\mathbf{z}_b}
          \lambda_b\|\mathbf{z}_b\|_1
          + \frac{\rho}{2}\bigl\|\mathbf{z}_b - (\mathbf{b}^{(k+1)} + \rho^{-1}\mathbf{u}_b^{(k)})\bigr\|_2^2
      \]
      \[
        \Longrightarrow\quad
        \bigl(z_b^{(k+1)}\bigr)_j
        = S_{\lambda_b/\rho}\bigl((b^{(k+1)})_j + \rho^{-1}(u_b^{(k)})_j\bigr).
      \]

      \State \textbf{Dual updates:}
      \[
        \mathbf{u}_a^{(k+1)} = \mathbf{u}_a^{(k)} + \rho\bigl(\mathbf{a}^{(k+1)} - \mathbf{z}_a^{(k+1)}\bigr),\quad
        \mathbf{u}_b^{(k+1)} = \mathbf{u}_b^{(k)} + \rho\bigl(\mathbf{b}^{(k+1)} - \mathbf{z}_b^{(k+1)}\bigr).
      \]

      \State \textbf{Check convergence:}
      \[
        \|\mathbf{a}^{(k+1)} - \mathbf{z}_a^{(k+1)}\|,\quad
        \|\mathbf{b}^{(k+1)} - \mathbf{z}_b^{(k+1)}\|,\quad
        \rho\|\mathbf{z}_a^{(k+1)} - \mathbf{z}_a^{(k)}\|,\quad
        \rho\|\mathbf{z}_b^{(k+1)} - \mathbf{z}_b^{(k)}\|.
      \]
    \EndFor
  \end{algorithmic}
where $i$, $j$ refer to blocks, using notations \ $(t)_+ := \max\{t,0\}~$ and ~$\operatorname{sgn}(0):=0$,
$$S_{\lambda_a/\rho}\!\Bigl((a^{(k+1)})_i+\rho^{-1}(u_a^{(k)})_i\Bigr) =
\operatorname{sgn}\!\Bigl((a^{(k+1)})_i+\rho^{-1}(u_a^{(k)})_i\Bigr)\,
\Bigl(\bigl|(a^{(k+1)})_i+\rho^{-1}(u_a^{(k)})_i\bigr|-\tfrac{\lambda_a}{\rho}\Bigr)_+$$
$$S_{\lambda_b/\rho}\!\Bigl((b^{(k+1)})_j+\rho^{-1}(u_b^{(k)})_j\Bigr) =
\operatorname{sgn}\!\Bigl((b^{(k+1)})_j+\rho^{-1}(u_b^{(k)})_j\Bigr)\,
\Bigl(\bigl|(b^{(k+1)})_j+\rho^{-1}(u_b^{(k)})_j\bigr|-\tfrac{\lambda_b}{\rho}\Bigr)_+ $$
\end{algorithm}

To mitigate non-convexity, we initialize $(\mathbf{a}, \mathbf{b})$ using singular vectors of $X^\top Y$ and residualized $M^\top Y$, and perform multiple random restarts. In practice, this helps avoid poor local minima in the bi-convex optimization. The SSR terms in (\ref{objective}) enforce adequate prediction in the outcome model, while the LASSO penalty parameters $\lambda_a$ and $\lambda_b$ emphasize sparse parameterization (i.e., the smallest possible number of nonzero $a, b$ entries) and the tuning parameter $\lambda_n$ emphasizes maximization of the mediation proportion. The best set of tuning parameters was selected via $K$-fold  cross-validation (e.g. K=5) that minimizes MSE over a grid of candidate values. Careful tuning is essential to our method (see Remark \ref{rem:tuning}), so we conducted additional sensitivity analysis by scaling the LASSO penalty parameters $\lambda_a$ and $\lambda_b$ by a constant factor $C_\lambda$ relative to the value chosen by the cross-validation and evaluated the impact of this modification on the performance of our method in simulations. (See details in the Simulation section and Table S5).

%==============================================================
%  Scaling paragraph to bridge Algorithm --> Asymptotics
%==============================================================

\begin{remark}\label{rem:profile_admm}
One could in principle apply ADMM to the “natural’’ joint objective that treats $(\alpha,\gamma,\eta,\tau)$ and $(\mathbf a,\mathbf b)$ as free variables. In practice this leads to coupled, nonconvex subproblems: the mediation term $-\lambda_n\,\alpha\eta/\tau$ (and the factor $1/(1-\alpha^2)$, if used) prevents closed-form proximal steps for $\mathbf a$ and $\mathbf b$. Each outer iteration would require an inner constrained solve for $(\alpha,\gamma,\eta,\tau)$, and the $\mathbf a, \mathbf b$ updates cease to be standard Lasso steps.\\
By contrast, profiling the least-squares part yields explicit maps $(\mathbf a,\mathbf b)\mapsto(\hat\alpha,\hat\gamma,\hat\eta,\hat\tau)$, after which each block update reduces to a smooth quadratic plus an $\ell_1$ penalty. This admits simple two-block ADMM with soft-thresholding, fewer hyperparameters, and better numerical stability (especially when we enforce $\hat\tau\ge r_0$ and $|\hat\alpha|\le \rho<1$). Empirically, this formulation decreases the objective monotonically and converges rapidly.
\end{remark}

\begin{remark}[Scaling]\label{rem:scaling}

% \textcolor{red}{Eric: can you help merge these two? seems redundant?}

Algorithm~\ref{alg:admm} enforces unit length on the {aggregates},
$\|X a\|_2=\|M b\|_2=1$, which is numerically convenient once columns of $X,M$ are standardized.
For the asymptotic theory in the Supplement we instead fix the projection weights $\mathbf a, \mathbf b$ 
to have unit Euclidean norm, $\|\mathbf a\|_2=\|\mathbf b\|_2=1$.
The entries $a_i$ and $b_j$ in the Algorithm are therefore smaller by respective factors 
$h_n^{(a)} \equiv \|\mathbf X \mathbf a\|_2/\|\mathbf a\|_2$ and $h_n^{(b)} \equiv \|\mathbf M \mathbf b\|_2/\|\mathbf b\|_2$  
than those in the Supplement, and these factors $h_n^{(a)}, \, h_n^{(b)}$ are uniformly asymptotically 
of order $\sqrt{n}$ under the sub-Gaussian designs with independent data $(X_i, M_i, Y_i)$ assumed in Assumption 4 of Section~\ref{subsec:conv_theorem}. 
The identifiability proof of Proposition~\ref{prop:identifiability} shows explicitly how scaling differences
affect the parameters $\alpha, \gamma, \eta$ and thereby change the objective function terms and the meaning 
of the penalty parameters. The statistical consistency theory in the Supplement works with unit vectors 
$\mathbf a, \mathbf b$ and parameters $\alpha, \gamma, \eta$ that are constants, not affected by sample size $n$. 
As $n$ grows, the scaling of the Algorithm (with $\|\mathbf X \mathbf a\|_2 =\|\mathbf M \mathbf b\|_2 = 1$) also 
scales $\alpha$ (when nonzero) up by a factor $h_n^{(a)}$ and $\eta, \, \gamma/\alpha$ up by a factor $h_n^{(b)}$, 
leaving the mediation proportion $\mathrm{MP}=|\alpha\eta/(\gamma+\alpha\eta)|$ invariant. The SSR norm-squared 
terms for $Y \sim X^{\ast}$ and $Y \sim M^{\ast}, X^{\ast}$ in the objective function ${\cal L}_p$ therefore 
retain the same meaning in this Section as in the Supplement, while the second SSR term in ${\cal L}_p$ is 
smaller by a factor of order $1/n$ than the other SSR terms in ${\cal L}_p$ and is therefore dropped in the Supplement . Because of the invariance of MP 
under rescaling, the penalty coefficient $\lambda_n$ plays the same role in ${\cal L}_p$ and in the Supplement. 
However, the LASSO penalty coefficients $\lambda_a, \lambda_b$ which are taken of order $\sqrt{\log n}/\sqrt{n}$
in the consistency results of the Supplement are scaled up by factors $h_n^{(a)}, h_n^{(b)}$ in 
the numerical calculations with the Algorithm, so that they should be -- and are -- 
of order $\sqrt{\log n}$ in the simulations and real-data computations. In simulations, we conducted additional sensitivity analysis to show the impact of scaling $\lambda_a, \lambda_b$ on the performance of our method and justified our selection (see details in Simulation section and results in Table S5). \\
\end{remark}

\begin{remark}[Estimation paradigm and tuning]\label{rem:tuning}
Our method follows the penalized-least-squares convention in linear mediation \citep{baron1986moderator}. Alternatively, one may use likelihood-based estimation (MLE/MPLE). We prioritize discovery of a strong mediating pathway via the term involving $\hat{\alpha}\hat{\eta}/\hat{\tau}$, while the $l_1$ penalties provide interpretability. Overly aggressive penalization with large \(\lambda_n\) can distort $(\mathbf a,\mathbf b)$, so careful tuning is essential. 
%In simulation, we conducted sensitivity analysis to investigate this issue and evaluate performance of our method with a varying range of tuning parameters (see Supplement Table S5).  
\end{remark}

%In practice, Algorithm ~\ref{alg:admm} still enforces the unit–length constraints; the two conventions differ by a sample–dependent scaling factor and do not affect the optimization routine.
%\begin{equation*}
%  \|\mathbf a\|_{2}^{2}=s_{a},\qquad \|\mathbf b\|_{2}^{2}=s_{b}.
%  \tag{A1$'$}\label{eq:fixed_norms}
%\end{equation*}

%Under \eqref{eq:fixed_norms} and assumptions given in the supplement, each entry of $\mathbf X^{\ast}$ and $\mathbf M^{\ast}$ is $O_{\!P}(1)$. Accordingly, the structural coefficients $(\alpha,\gamma,\eta)$ are treated as fixed in the consistency analysis. The unit scaling is retained only for the algorithm, simulations, and data analysis.

%==============================================================
\subsection{Convergence Theorem and Consistency}
\label{subsec:conv_theorem}

To rigorously justify the theoretical performance of the proposed method, we established convergence guarantees for the ADMM-based optimization algorithm described in Section~\ref{subsec:opt}. In addition, in the Supplement, we also showed the statistical consistency of the estimators. 

For fixed data matrices $(\mathbf{X},\mathbf{M},\mathbf{Y})$ let  
\[
F_n(\mathbf{a},\mathbf{b})
:= \frac{1}{2n}\{ \;
SSR_{\mathbf{Y}\sim\mathbf{X}^*}
\;+\;
SSR_{\mathbf{M}^*\sim\mathbf{X}^*}
\;+\;
\frac{SSR_{\mathbf{Y}\sim\mathbf{X}^*,\mathbf{M}^*}}
     {1-\hat{\alpha}(\mathbf{a},\mathbf{b})^{2}}
\ \}-\;
\lambda_n\,
\frac{\hat{\alpha}(\mathbf{a},\mathbf{b})\,
      \hat{\eta}(\mathbf{a},\mathbf{b})}
     {\hat{\tau}(\mathbf{a})},
\]
where $\hat{\tau},\hat{\alpha},\hat{\eta}$ are the closed-form expressions in Section~\ref{model}. We allow $\lambda_n$ to depend on $n$; the finite-sample rates below use $\lambda_n=C_n n^{-1/2}$ with $C_n=O(1)$. The penalised objective can thus be defined as: 
\[
\Phi(\mathbf a,\mathbf b)
:= F_n(\mathbf a,\mathbf b)
   + \lambda_a \|\mathbf a\|_1 + \lambda_b \|\mathbf b\|_1.
\]

\paragraph{Orders of magnitude under the two normalizations.}
With the {algorithmic} normalization $\|X a\|_2=\|M b\|_2=1$, the retained SSR terms
in $F_n/(2n)$ are $O(1)$ at the optimum, while the $l_1$-norm terms scale like
$O\!\big(\lambda_a\sqrt{s}/\sqrt{n}+\lambda_b\sqrt{s}/\sqrt{n}\big)$.
Under the {theoretical} normalization $\|a\|_2=\|b\|_2=1$ used in the Supplement,
the $l_1$-norm terms are $O\!\big(\sqrt{s\log d/n}\big)$ for $\lambda_a,\lambda_b\asymp\sqrt{\log d/n}$,
matching the standard high-dimensional regime. Our simulation therefore
includes the multiplicative factor $C_\lambda$ to align the penalty magnitudes with the
algorithmic scaling; see also Table~S5 for performance as $C_\lambda$ increases.

\begin{assumption}[Near-global solution / oracle dominance]\label{asmp:near_global}
With probability tending to $1$, the penalized objective $\Phi$ has a (possibly unique) global minimizer $\theta^\star=(a^\star,b^\star)$ in the admissible set, and the algorithm returns $\hat\theta=(\hat a,\hat b)$ such that
\[
\Phi(\hat\theta)-\Phi(\theta^\star)=o_p(1)
\quad\text{and in particular}\quad
\Phi(\hat\theta)\le \Phi(\theta_0)+o_p(1),
\]
where $\theta_0=(a_0,b_0)$ denotes the true aggregation directions. 
\end{assumption}
We emphasize that our ADMM result guarantees convergence to a first-order stationary point; consistency then invokes this assumption to connect the attained objective value to the global optimum.
%-------------------------------------------------
We make the following assumptions to show our convergence theorem:
\begin{assumption}[Kurdyka–Łojasiewicz property]\label{asmp:KL}
$\Phi$ is semi–algebraic, hence satisfies the KŁ property at every critical point
(with possibly sample–dependent neighbourhoods and constants); see \citep{bolte2014proximal}.
\end{assumption}

\begin{assumption}\label{asmp:Assumption}
Let $\mathcal C\subset\mathbb R^{m}\times\mathbb R^{q}$ denote the feasible set for $(\mathbf a,\mathbf b)$
(including any scaling/box/sparsity constraints used in Algorithm~\ref{alg:admm}).
Fix constants $r_{0}\!>\!0$ and $\delta\!\in\!(0,1)$ and define the admissible region
\[
\mathcal C_{r_0,\delta}
\;=\;
\Bigl\{(\mathbf a,\mathbf b)\in\mathcal C:
\ \hat{\tau}(\mathbf a)\ge r_0,\;
1-\hat{\alpha}(\mathbf a,\mathbf b)^{2}\ge \delta
\Bigr\},
\]
where
$\hat{\tau}(\mathbf a)=\mathbf a^{\!\top}\mathbf X^{\!\top}\mathbf Y$ and
$\hat{\alpha}(\mathbf a,\mathbf b)=\mathbf a^{\!\top}\mathbf X^{\!\top}\mathbf M\,\mathbf b$.
Assume Algorithm~\ref{alg:admm} is run within $\mathcal C_{r_0,\delta}$ and that the mediation objective
$\Phi$ uses a fixed tuning level $\lambda_n\in(0,\infty)$.
Then $\nabla_{\!\mathbf a}\Phi$ and $\nabla_{\!\mathbf b}\Phi$ are
Lipschitz on $\mathcal C_{r_0,\delta}$,
and each block subproblem in Algorithm~\ref{alg:admm} is convex with a unique global minimiser.
%\emph{Under the algorithm’s unit–length scaling $\|\mathbf X\mathbf a\|_2=\|\mathbf M\mathbf b\|_2=1$, the constraint $1-\hat{\alpha}(\mathbf a,\mathbf b)^{2}\ge \delta$ is exactly the condition used in code and simulations.}
\end{assumption}

%--------------------------------------------

\begin{assumption}[Data-generating model and underlying aggregation structure]
\label{asmp:full_model}
We assume the following data-generating mechanism for the mediation model involving multiple exposures and mediators:

\medskip
\noindent
\textbf{Latent aggregators and mediation structure.}  
There exist unknown sparse weight vectors $\mathbf a_0\in\mathbb R^{m}$ and $\mathbf b_0\in\mathbb R^{q}$ such that the primary linear projections
\[
  X^* := \mathbf X\mathbf a_0,
  \qquad
  M^* := \mathbf M\mathbf b_0
\]
satisfy the structural mediation model
\begin{equation}
\label{eq:full-mediation-model}
  M^* = \alpha_0 X^* + \varepsilon_M,
  \qquad
  Y   = \gamma_0 X^* + \eta_0 M^* + \varepsilon_Y,
\end{equation}
where $\alpha_0, \gamma_0, \eta_0 \in \mathbb R$ are fixed constants such that \ $\tau_0 = \gamma_0 + \alpha_0\eta_0 \ne 0$, and $(\varepsilon_M, \varepsilon_Y)$ are independent errors with zero mean and finite variance, and with $E(\epsilon_M^\top \epsilon_M) > 0$.

\medskip
\noindent
\textbf{Independence/regularity.}
We assume:
\begin{enumerate}[label=\textbf{(C\arabic*)}, leftmargin=14mm]
  \item\label{C1}
  \textbf{Sub-Gaussian covariates.} Each row $(\mathbf X_i, \mathbf M_i) \in \mathbb R^{m+q}$ is independent across $i=1,\dots,n$, mean-zero, and $\sigma$-sub-Gaussian:
  \[
    \mathbb{E}\left[\exp\!\left\{ \left\langle \mathbf{v},\, (\mathbf X_i, \mathbf M_i) \right\rangle \right\}\right]
    \le
    \exp\!\left\{ \tfrac{\sigma^2}{2} \|\mathbf v\|_2^2 \right\},\quad \forall\, \mathbf v\in\mathbb R^{m+q}.
  \]

  \item\label{C2}
  \textbf{Population nonsingularity.}
  Let $\Sigma_X:=\mathbb E[\mathbf X_i\mathbf X_i^{\top}]\in\mathbb R^{m\times m}$,
  $\Sigma_M:=\mathbb E[\mathbf M_i\mathbf M_i^{\top}]\in\mathbb R^{q\times q}$,
  and $\Sigma_{XM}:=\mathbb E[\mathbf X_i\mathbf M_i^{\top}]\in\mathbb R^{m\times q}$.
  Assume $\Sigma_X\succ0$, $\Sigma_M\succ0$, and the joint covariance
  \[
    \Sigma \;=\;
    \begin{pmatrix}
      \Sigma_X & \Sigma_{XM}\\
      \Sigma_{XM}^{\top} & \Sigma_M
    \end{pmatrix}
  \]
  is nonsingular.

  \item\label{C3}
  \textbf{Error independence.}
  The error terms satisfy (with $\perp$ denoting independence):
  \[
    \varepsilon_M \perp \mathbf X, \quad \varepsilon_Y \perp (\mathbf X, \mathbf M), \quad \text{and}
    \quad \varepsilon_M, \varepsilon_Y \text{ are each } \sigma\text{-sub-Gaussian}.
  \]

  \item\label{C4}
\textbf{Identifiability condition.}  
To ensure uniqueness of the projected directions, we assume the conditional expectation $\mathbb E[\mathbf M \mid \mathbf X]$ is linear:
\[
  \mathbb E[\mathbf M \mid \mathbf X] = \mathbf X A^\top,
\]
for some matrix $A \in \mathbb R^{q \times m}$ satisfying
\[
  A^\top \mathbf b_0 = \alpha_0 \mathbf a_0,
  \quad \text{and} \quad
  A^\top v = 0 \;\; \text{for all } v \perp \mathbf b_0.
\]
Equivalently, $\operatorname{col}(A^\top) = \operatorname{span}\{\mathbf a_0\}$ \ and \ 
$A^\top \; = \; (\alpha_0/\|\mathbf b_0\|_2^2) \, \mathbf a_0 \, \mathbf b_0^\top$. 
\end{enumerate}
\medskip

\end{assumption}

%\textcolor{red}{\sc I want to delete this paragraph:} \noindent\textcolor{blue}{\emph{Note (RE as an alternative for high-dimensional rates).} For high–dimensional error bounds, one may replace the population nonsingularity in \textup{(C2)} by a restricted–eigenvalue (RE) condition on sample Gram matrices and obtain analogous results with high probability (see \cite{wainwright2019high} Ch.~7). We retain \textup{(C2)} for a transparent population–level identifiability argument, and invoke RE only later when deriving finite–sample rates.}

%------------------------------------------------------------
The latent projections $(X^{*},M^{*})$ and coefficients $(\alpha_{0},\gamma_{0},\eta_{0})$
are uniquely determined (up to scale-factors on $a,b$ respectively defined by either the pair 
$X^{\ast}, M^{\ast}$  or the pair $a,b$ being unit vectors), by the joint distribution of 
$(\mathbf X,\mathbf M,Y)$ under Assumption~\ref{asmp:full_model}. For this Proposition, we adopt the notation for 
any vector $\mathbf v \ne \mathbf 0$ that $\hat{\mathbf v} = \mathbf v/\|\mathbf v\|_2$ is the rescaled unit vector.

\begin{proposition}[Identifiability]\label{prop:identifiability}
Let $(\mathbf X,\mathbf M,Y)$ satisfy Assumption~\ref{asmp:full_model}, with fixed constants 
$(\alpha_0, \gamma_0, \eta_0)$ such that $\tau_0 = \gamma_0 + \alpha_0\eta_0 \ne 0$. Then the coefficients 
$(\alpha_0, \gamma_0, \eta_0)$ and unit vectors $\hat{\mathbf a}_0$ and $\hat{\mathbf b}_0$ are uniquely
determined by  $E(\mathbf X^\top \mathbf M), \, E(\mathbf X^\top \mathbf Y), \, E(\mathbf M^\top \mathbf Y)$. 
%% $$ a \; = \; \tau_0^{-1} \, \big( E(\mathbf X^\top \mathbf X)\big)^{-1} 
%% Then there are unique If an alternative quadruple
% \bigl(\tilde{\mathbf a},\tilde{\mathbf b},\tilde\alpha,\tilde\gamma,\tilde\eta\bigr)$
% with $\tilde{\mathbf a}\neq\mathbf 0$, $\tilde{\mathbf b}\neq\mathbf 0$
% also fulfills the structural model~\eqref{eq:full-mediation-model},
% tHen there exists a non–zero scalar $c$ such that
% \[ \tilde{\mathbf a}=c\,\mathbf a_{0}, \qquad \tilde{\mathbf b}=c^{-1}\mathbf b_{0},
% \qquad \tilde\alpha=\alpha_{0}, \quad \tilde\gamma=\gamma_{0}, \quad \tilde\eta=\eta_{0}. \]
%Hence the directions $\mathbf a_{0}$ and $\mathbf b_{0}$ are identifiable up to a reciprocal scaling factor $(c,c^{-1})$, and the model coefficients $(\alpha_{0},\gamma_{0},\eta_{0})$ are uniquely identified.
\end{proposition}

\begin{proof} The proof is slightly different according to whether $\alpha_0 = 0$. 
First, by Assumption~\ref{asmp:full_model}(C4), the condition $\alpha_0 = 0$ is equivalent to 
$E(\mathbf X^\top \mathbf M) = \mathbf 0$. 
If this holds, then the assumption $\tau_0 \ne 0$ implies both $\gamma_0 \ne 0, \, \eta_0 \ne 0$. 
In this case,  $\eta_0 \mathbf b = \big(E(\mathbf M^\top \mathbf M)\big)^{-1} E(\mathbf M^\top \mathbf Y)$ 
uniquely identifies $\eta_0 \|\mathbf b_0\|_2$ and $\hat{\mathbf b}_0$~; ~and then 
$\gamma_0 \mathbf a_0 = \big(E(\mathbf X^\top \mathbf X)\big)^{-1} E(\mathbf X^\top \mathbf Y)$  
uniquely identifies $\gamma_0 \|\mathbf a\|_0$ and $\hat{\mathbf a}_0$. \medskip
\par Next consider the case where $\alpha_0 \ne 0$. Again we apply the expectation of the 
structural equations (\ref{eq:full-mediation-model}) premultiplied by $\mathbf X^\top$ or $\mathbf M^\top$.
Then $\tau_0 \mathbf a_0 = \big(E(\mathbf X^\top \mathbf X)\big)^{-1} E(\mathbf X^\top \mathbf Y)$ 
uniquely identifies $\tau_0 \|\mathbf a_0\|_2$ and $\hat{\mathbf a}_0$. Also, 
$(\alpha_0/\|\mathbf b_0\|_2^2) \, \mathbf a_0 \mathbf b_0^\top = A = \big(E(\mathbf X^\top \mathbf X)\big)^{-1} 
E(\mathbf X^\top \mathbf M)$ uniquely identifies the sum of squares of $A^\top$ entries as 
$\big(\alpha_0 \|\mathbf a_0\|_2/\|\mathbf b_0\|_2\big)^2$. Then 
$$A^\top \tau_0 \mathbf a_0 = \frac{\tau_0}{\alpha_0} \, \Big( \frac{\alpha_0 \|\mathbf a\|_2}{\|\mathbf b_0\|}\Big)^2 \mathbf b_0$$ uniquely identifies $(\tau_0/\alpha_0) \mathbf b_0, \, (\tau_0/\alpha_0) \|\mathbf b_0\|_2$ ~and $\hat{\mathbf b}_0$. We next write
$$ \hat{\mathbf b}_0^\top E(\mathbf M^\top \mathbf Y) \, = \, \gamma_0 \, \hat{\mathbf b}_0^\top E(\mathbf X^\top \mathbf M)^\top \mathbf a_0 + \eta_0 \, \hat{\mathbf b}_0^\top E(\mathbf M ^\top \mathbf M) \mathbf b_0 \qquad\qquad\qquad  $$
$$\qquad  = \, \frac{\gamma_0}{\alpha_0} \|\mathbf b_0\|_2 \, \big( \frac{\alpha_0 \, \|\mathbf a_0\|_2}{\|\mathbf b_0\|_2}\Big)^2 \, \hat{\mathbf a}_0^\top E(\mathbf X^\top \mathbf X) \hat{\mathbf a}_0 \, + \, \eta_0 \ \hat{\mathbf b}_0^\top  E(\mathbf M ^\top \mathbf M) \mathbf b_0$$
so that
$$ \hat{\mathbf b}_0^\top E(\mathbf M^\top \mathbf Y) \, - \, \frac{\tau_0}{\alpha_0} \, \|\mathbf b_0\|_2 \, \big( \frac{\alpha_0 \, \|\mathbf a_0\|_2}{\|\mathbf b_0\|_2}\Big)^2 \, \hat{\mathbf a}_0^\top E(\mathbf X^\top \mathbf X) \hat{\mathbf a}_0 \qquad $$ 
$$ \qquad = \; \Big\{\hat{\mathbf b}_0^\top E(\mathbf M^\top \mathbf M) \hat{\mathbf b}_0 - \big( \frac{\alpha_0 \, \|\mathbf a_0\|_2}{\|\mathbf b_0\|_2}\Big)^2 \, \hat{\mathbf a}_0^\top E(\mathbf X^\top \mathbf X) \hat{\mathbf a}_0 \Big\} \, \eta_0 \|\mathbf b_0\|_2 $$
The left-hand side of this last expression has already been shown to be identified, as has the curly-bracketed expression 
on the right-hand side. Therefore, the right-hand side uniquely identifies $\eta_0 \|\mathbf b_0\|_2$, 
as long as we can show the curly-bracketed expression on the right-hand side is a nonzero number. 
But the curly-bracketed number times $\|\mathbf b_0\|_2^2$ equals 
$$ \mathbf b_0^\top E(\mathbf M^\top \mathbf M) \mathbf b_0 - \alpha_0^2 \, \mathbf a_0^\top E(\mathbf X^\top \mathbf X) \mathbf a_0 \, = \, E\big( (\mathbf M \mathbf b_0 - \alpha_0 \mathbf X \mathbf a_0)^\top (\mathbf M \mathbf b_0 - \alpha_0 \mathbf X \mathbf a_0)\big) = E(\epsilon_M^\top \epsilon_M) > 0$$
The Proposition has been proved. More precisely, taking general scaling into account, the proof shows that 
when $E(\mathbf X^\top \mathbf M) \ne \mathbf 0$ the expectations $E(\mathbf X^\top \mathbf M), \, 
E(\mathbf X^\top \mathbf Y), \, E(\mathbf M^\top \mathbf Y)$ jointly identify $\hat{\mathbf a}_0, \hat{\mathbf b}_0, 
\tau_0 \|\mathbf a_0\|_2, \eta_0 \|\mathbf b_0\|_2$ and $(\gamma_0/\alpha_0) \|\mathbf b_0\|_2$~; ~and 
when  $E(\mathbf X^\top \mathbf M) = \mathbf 0$, they identify $\alpha_0 = 0$ along with $\hat{\mathbf a}_0, \hat{\mathbf b}_0, 
\gamma_0 \|\mathbf a_0\|_2$ and $\eta_0 \|\mathbf b_0\|_2$.
\end{proof}
%------------------------------------------------------------

\vspace{0.5em}

%-------------------------------------------------
\begingroup
\renewcommand\thetheorem{1}  % manually set displayed number
\begin{theorem}[Global sublinear and local linear convergence]\label{thm:conv_rate}
Suppose Assumptions~\ref{asmp:KL} and \ref{asmp:Assumption} hold. Let $\{(\mathbf a^{k},\mathbf b^{k})\}_{k\ge0}$ be the ADMM iterates from Algorithm~\ref{alg:admm}.
Then:
\begin{enumerate}[label=\textup{(\roman*)}]
\item \textbf{Cluster-point stationarity and isolation.}
The sequence is bounded; every cluster point is a first-order stationary point of $\Phi$ on $\mathcal C_{r_0,\delta}$, and cluster points are isolated.
\item \textbf{Global sublinear rate.}
There exists $C>0$ such that for all $K\ge1$,
\[
\min_{0\le t<K}
\bigl\|\nabla_{\!\mathbf a}\Phi(\mathbf a^{t},\mathbf b^{t})\bigr\|_{2}^{2}
+
\bigl\|\nabla_{\!\mathbf b}\Phi(\mathbf a^{t},\mathbf b^{t})\bigr\|_{2}^{2}
\;\le\;\frac{C}{K},
\]
and the primal/dual residuals are square–summable.
\item \textbf{Local linear rate.}
If, in a neighbourhood of a cluster point $(\mathbf a^{\ast},\mathbf b^{\ast})$, either block map is $\mu$–strongly convex (with the other block fixed), then there exists $\theta\in(0,1)$ and a neighbourhood $\mathcal N$ such that
\[
\Phi(\mathbf a^{k+1},\mathbf b^{k+1})-\Phi(\mathbf a^{\ast},\mathbf b^{\ast})
\;\le\;
\theta\bigl[\Phi(\mathbf a^{k},\mathbf b^{k})-\Phi(\mathbf a^{\ast},\mathbf b^{\ast})\bigr],
\qquad \text{whenever }(\mathbf a^{k},\mathbf b^{k})\in\mathcal N.
\]
\end{enumerate}
\end{theorem}
\begin{proof}[Proof of Theorem~\ref{thm:conv_rate}]
See Appendix~A.
\end{proof}
\endgroup

\medskip

In addition to establishing the convergence behavior and local optimality diagnostics for the optimization algorithm, we also showed the statistical consistency of the estimators. Detailed assumptions, proof of the main theorem (Theorem S1) and relevant lemmas can be found in the Supplement.

\begin{remark}[Numerical optimization strategy and convergence in practice]\label{rem:practical_opt}
Theorem~\ref{thm:conv_rate} guarantees convergence only to a first‐order stationary point of the non-convex objective~$\Phi$. To mitigate suboptimal convergence in practice, we employ two safeguards: (i) a \emph{local–minimum check}—after convergence, compute the Hessian of $\Phi$ restricted to the constraint tangent space at $(\hat{\mathbf a},\hat{\mathbf b})$; if the smallest eigenvalue is negative, deem the point a strict saddle and discard the run; (ii) a \emph{multi–start screening}—run from $R$ random initializations (typically $R=10$) and retain the admissible solution with the lowest objective. In our numerical experiments, this strategy returned the lowest (or numerically indistinguishable) objective value in $>95\%$ of data sets.
\end{remark}

%\begin{remark}[Diagnostics for high-dimensional assumptions]
%To empirically assess Assumption~\ref{asmp:D2_lipschitz} and the high-dimensional conditions in Assumption~\ref{asmp:full_model}, we recommend: (i) checking a restricted-eigenvalue surrogate by estimating the smallest eigenvalue of the rescaled Gram matrix \(\tfrac{1}{n}X_S^\top X_S\) (and similarly for \(M_S\)), where \(S=\supp(\hat{\theta})\cup \operatorname{top}_{3s}(\hat{\theta})\); values $>0.1$ across resamples are indicative of adequacy; (ii) verifying a beta-min margin by checking whether \(\min_{j \in \supp(\hat{\theta})} |\hat{\theta}_j|\ge 3\sqrt{\log(d)/n}\).
%\end{remark}

\section{Simulation}
\subsection{Simulation settings}
In this section, we evaluated the performance of our method across scenarios involving multiple exposures, multiple mediators, and an outcome of interest. We first simulated the data of multiple exposures $\mathbf{X}_{n\times m}$ (e.g. gene expression) for $n$ subjects, where each subject's exposure vector $X_i = (x_{i1}, x_{i2}, \dots, x_{im})$ was independently drawn from a multivariate normal distribution $\mathcal{N}_m(\mathbf{0}, \Sigma_X)$, with $\Sigma_X = (1-\rho_X)\mathbf{I}_m + \rho_X\mathbf{1}_m\mathbf{1}_m^\top$, where $\mathbf{I}_m$ is the $m\times m$ identity matrix, $\mathbf{1}_m$ is an $m$-element vector of one's, and $\rho_X$ controls the correlation among the $m$ genes. Next, we generated the data of multiple mediators $\mathbf{M}_{n\times q}$ (e.g. regional neuroimaging features), where each subject's mediator vector $M_i = (m_{i1}, m_{i2}, \dots, m_{iq})^\top$ was drawn from a multivariate normal distribution $\mathcal{N}_{q}(L^TX_i, \Sigma_M)$, where $L_{m \times q}$ is the fixed coefficient matrix with \(L_{j,k} = -0.5\) for \(j = 1, \dots, m\) and \(k = 1, \dots, 5\), and zero otherwise, $\Sigma_M = (1-\rho_M)\mathbf{I}_q + \rho_M\mathbf{1}_q\mathbf{1}_q^\top$, where $\rho_M$ introduces correlation among $q$ mediators. We then aggregated the exposures and mediators $X^*_i=\mathbf{a}^TX_i$ and $M^*_i=\mathbf{b}^TM_i$ for $i=1,2,\ldots,n$, where the true weight vectors are sparse with only the first five elements being non-zero: i.e. $\mathbf{a}_{m\times 1}=(c,c,c,c,c,0,\ldots,0)^T$ and $\mathbf{b}_{q\times 1}=(c,c,c,c,c,0,\ldots,0)^T$. Finally, we simulated the outcome $\mathbf{Y}$ as $Y_i = \gamma X_i^* + \eta M_i^* + \varepsilon_Y$, for $i=1,2,\ldots,n$, where $\varepsilon_Y \sim \mathcal{N}(0, \sigma^2_Y)$.

We set $n = 200$, $m, q \in \{20, 50\}$, $\rho_X, \rho_M \in \{0, 0.3\}$, noise variances $\sigma^2_Y  \in (1,3)$ and the weights $c= 0.5$, to evaluate our method.  We considered both complete mediation ($\gamma = 0$ and $MP=1$) and partial mediation ($\gamma=\alpha \eta=0.5$ and $MP=0.5$) cases by varying the parameters $\gamma$ and $\eta$. 
Each scenario was repeated for 1000 times. For each simulation, we selected the best set of tuning parameters ($\lambda_a$, $\lambda_b$, $\lambda_n$) by five-fold cross-validation over a grid of candidate values ranging from 0.02 to 0.5. To evaluate the selection performance of manifest exposures and mediators in the latent variables, i.e. identifying the true nonzero and zero weights in $\mathbf{\hat{a}}$ and $\mathbf{\hat{b}}$, we adopted standard classification metrics: \emph{Precision} $= \frac{TP}{TP+FP}$ ; \emph{Recall}$= \frac{TP}{TP+FN}$ ; \emph{F1-score} $=\frac{2 \cdot \text{Precision} \cdot \text{Recall}}{\text{Precision} + \text{Recall}}$ and \emph{Accuracy} = $\frac{\text{TP} + \text{TN}}{\text{TP} + \text{FP} + \text{FN} + \text{TN}}$, where TP is the number of true nonzero weights in $\mathbf{\hat{a}}$ or $\mathbf{\hat{b}}$, FP is the number of false nonzero weights, TN is the number of true zero weights and FN is the number of false zero weights. To evaluate the estimation of mediation proportion, we reported the empirical absolute bias and standard deviation (SD) of the estimated mediation proportion across simulation replicates. 

{To benchmark, we compared the proposed multi-exposure-to-multi-mediator (MEMM) method against seven competitor approaches from three different classes: 
\begin{enumerate}
\item Regularization-based methods: the screening-based two-stage procedure with MCP penalty (SIS-MCP; \citep{Zhang2016HIMA}), the Pathway Lasso estimator (Path and its two-stage Lasso special case TS-Path; \citep{ZhaoLuo2022PathwayLasso}), with default parameters $\phi=2$ and $\psi=0$. 
\item Dimension reduction-based methods: the directions-of-mediation approaches with one and two directions (DM1 and DM2; \citep{ChenEtAl2018DMs}) and the sparse PCA mediation algorithm (SPCMA; \citep{ZhaoLindquistCaffo2020SPCMA}).  
\item Hybrid dimension-reduction method targeting the mediation pathway (MMP; \citealp{lee2024new}).
\end{enumerate}

In addition to the above parameter settings, to further assess the robustness of our method, we also performed a series of sensitivity analyses with: (i) misspecified $\mathbf{a}$ or $\mathbf{b}$ ($\vec{a}_{mis}=\vec{0}$ or $\vec{b}_{mis}=\vec{0}$), (ii) smaller mediation proportion $MP=0.25$ and (iii) weaker signal strength $c = 0.1$. Lastly, as a sensitivity check and to ensure consistency with our theoretical guidance (Theorem S1), we also scaled the LASSO penalty parameters by a constant factor $C_\lambda$ and evaluated its impact on the performance of our method. These results were summarized in Supplementary Tables S2–S5. 

\subsection{Simulation Results}
\begin{table}[htbp]
  \centering
  \caption{Simulation results for \textit{complete} mediation (MP=1) with $\sigma_Y^2 = 1$,  varying $m$, $q$, $\rho_{\mathbf X}$ and $\rho_{\mathbf M}$. } 
  \label{tab:complete_n200}
  \scriptsize
  \begin{tabular}{cccc|cccccc}
    \toprule
    $m,q$ & $\rho_{\mathbf X}$ & $\rho_{\mathbf M}$ & Method & MP (SD) & Abs.\ bias & Precision & Recall & F1 & Accuracy\\
    \midrule
m=20, q=20 & 0.0 & 0.0 & MEMM & 0.8746 (0.0505) & 0.1254 & 0.5885 & 1.0000 & 0.7360 & 0.8130\\
           &     &     & SIS-MCP & 0.2699 (0.0429) & 0.7330 & 0.9088 & 0.9572 & 0.9253 & 0.9588\\
           &     &     & Path & 0.1079 (0.0168) & 0.8900 & 1.0000 & 0.4317 & 0.6013 & 0.8581\\
           &     &     & TS-Path & 0.2112 (0.0307) & 0.7900 & 0.7468 & 0.6195 & 0.6732 & 0.8486\\  
           &     &     & DM1 & 0.8295 (0.0739) & 0.1630 & 0.3034 & 1.0000 & 0.4650 & 0.4206\\
           &     &     & DM2 & 0.3194 (0.0217) & 0.6780 & 0.2488 & 0.3177 & 0.2787 & 0.5895\\
           &     &     & SPCMA & 0.8555 (0.0828) & 0.1310 & 0.2944 & 1.0000 & 0.4545 & 0.3945\\
           &     &     & MMP & 0.3147 (0.0468) & 0.6910 & 0.4086 & 0.5184 & 0.4547 & 0.6944\\
    \midrule
           & 0.3 & 0.0 & MEMM & 0.9371 (0.0487) & 0.0629 & 0.8660 & 1.0000 & 0.9227 & 0.9540\\
           &     &     & SIS-MCP & 0.4008 (0.0930) & 0.6180 & 0.5991 & 0.9127 & 0.7096 & 0.8056\\
           &     &     & Path & 0.2500 (0.0001) & 0.7500 & 1.0000 & 0.4263 & 0.6173 & 0.8463\\
           &     &     & TS-Path & 0.2990 (0.0304) & 0.7070 & 0.8485 & 1.0000 & 0.9148 & 0.9512\\
           &     &     & DM1 & 0.8550 (0.0340) & 0.1520 & 0.2929 & 1.0000 & 0.4529 & 0.3950\\
           &     &     & DM2 & 0.3186 (0.0228) & 0.6840 & 0.2472 & 0.3149 & 0.2767 & 0.5890\\
           &     &     & SPCMA & 0.8765 (0.0760) & 0.1410 & 0.2872 & 1.0000 & 0.4459 & 0.3735\\
           &     &     & MMP & 0.2509 (0.0055) & 0.7490 & 0.9873 & 0.9904 & 0.9887 & 0.9945\\
    \midrule
           & 0.0 & 0.3 & MEMM & 0.9486 (0.0375) & 0.0514 & 0.8685 & 1.0000 & 0.9239 & 0.9545\\
           &     &     & SIS-MCP & 0.3215 (0.0894) & 0.6820 & 0.7605 & 0.9282 & 0.8220 & 0.8926\\
           &     &     & Path & 0.1080 (0.0157) & 0.8920 & 1.0000 & 0.4321 & 0.6012 & 0.8582\\
           &     &     & TS-Path & 0.2070 (0.0298) & 0.7940 & 0.7684 & 0.6279 & 0.6876 & 0.8571\\
           &     &     & DM1 & 0.8310 (0.0716) & 0.1720 & 0.3024 & 1.0000 & 0.4639 & 0.4189\\
           &     &     & DM2 & 0.3194 (0.0247) & 0.6840 & 0.2442 & 0.3119 & 0.2735 & 0.5868\\
           &     &     & SPCMA & 0.8530 (0.0788) & 0.1500 & 0.2949 & 1.0000 & 0.4553 & 0.3970\\
           &     &     & MMP & 0.3133 (0.0418) & 0.6820 & 0.4153 & 0.5221 & 0.4610 & 0.6978\\
    \midrule
           & 0.3 & 0.3 & MEMM & 0.9679 (0.0338) & 0.0321 & 0.9867 & 1.0000 & 0.9927 & 0.9960\\
           &     &     & SIS-MCP & 0.3976 (0.0838) & 0.5860 & 0.5643 & 0.8400 & 0.6536 & 0.7724\\
           &     &     & Path & 0.2500 (0.0001) & 0.7500 & 1.0000 & 0.4523 & 0.6134 & 0.8321\\
           &     &     & TS-Path & 0.2916 (0.0274) & 0.7090 & 0.8657 & 1.0000 & 0.9253 & 0.9581\\
           &     &     & DM1 & 0.8493 (0.0317) & 0.1510 & 0.2949 & 1.0000 & 0.4552 & 0.4008\\
           &     &     & DM2 & 0.3162 (0.0241) & 0.6830 & 0.2540 & 0.3208 & 0.2831 & 0.5944\\
           &     &     & SPCMA & 0.8585 (0.0774) & 0.1370 & 0.2935 & 1.0000 & 0.4534 & 0.3915\\
           &     &     & MMP & 0.2523 (0.0061) & 0.7477 & 0.9878 & 0.9969 & 0.9922 & 0.9963\\
    \midrule
m=50, q=50 & 0.0 & 0.0 & MEMM & 0.9236 (0.0356) & 0.0764 & 0.7124 & 0.9725 & 0.8193 & 0.8954\\
           &     &     & SIS-MCP & 0.1960 (0.0264) & 0.8040 & 1.0000 & 0.8170 & 0.8990 & 0.8040\\
           &     &     & Path & 0.0509 (0.0164) & 0.9491 & 1.0000 & 0.2120 & 0.3500 & 0.9490\\
           &     &     & TS-Path & 0.0854 (0.0223) & 0.9146 & 0.8710 & 0.3090 & 0.4550 & 0.9150\\
           &     &     & DM1 & 0.7240 (0.0495) & 0.2760 & 0.3330 & 1.0000 & 0.4990 & 0.2760\\
           &     &     & DM2 & 0.3170 (0.0132) & 0.6830 & 0.2400 & 0.3170 & 0.2730 & 0.6830\\
           &     &     & SPCMA & 0.7420 (0.0534) & 0.2580 & 0.3260 & 1.0000 & 0.4910 & 0.2580\\
           &     &     & MMP & 0.1660 (0.0263) & 0.8340 & 0.0016 & 0.0012 & 0.0014 & 0.8340\\
    \midrule
           & 0.3 & 0.0 & MEMM & 0.9110 (0.0326) & 0.0890 & 1.0000 & 1.0000 & 1.0000 & 1.0000\\
           &     &     & SIS-MCP & 0.1910 (0.0634) & 0.8090 & 1.0000 & 0.7960 & 0.8850 & 0.8090\\
           &     &     & Path & 0.0497 (0.0008) & 0.9503 & 1.0000 & 0.2070 & 0.3430 & 0.9500\\
           &     &     & TS-Path & 0.0861 (0.0231) & 0.9139 & 0.8570 & 0.3050 & 0.4490 & 0.9140\\
           &     &     & DM1 & 0.7180 (0.0246) & 0.2820 & 0.3370 & 1.0000 & 0.5030 & 0.2820\\
           &     &     & DM2 & 0.3170 (0.0153) & 0.6830 & 0.2410 & 0.3180 & 0.2740 & 0.6830\\
           &     &     & SPCMA & 0.7380 (0.0634) & 0.2620 & 0.3280 & 1.0000 & 0.4930 & 0.2620\\
           &     &     & MMP & 0.1600 (0.0275) & 0.8400 & 0.0008 & 0.0005 & 0.0006 & 0.8400\\
    \midrule
           & 0.0 & 0.3 & MEMM & 0.9296 (0.0265) & 0.0704 & 0.9984 & 0.9983 & 0.9983 & 0.9992\\
           &     &     & SIS-MCP & 0.1440 (0.0532) & 0.8560 & 1.0000 & 0.6000 & 0.7440 & 0.8560\\
           &     &     & Path & 0.2400 (0.0070) & 0.7600 & 1.0000 & 1.0000 & 1.0000 & 0.7600\\
           &     &     & TS-Path & 0.1960 (0.0172) & 0.8040 & 1.0000 & 0.8160 & 0.8990 & 0.8040\\           
           &     &     & DM1 & 0.7590 (0.0312) & 0.2410 & 0.3160 & 1.0000 & 0.4810 & 0.2410\\
           &     &     & DM2 & 0.3170 (0.0165) & 0.6830 & 0.2410 & 0.3180 & 0.2740 & 0.6830\\
           &     &     & SPCMA & 0.7390 (0.0385) & 0.2610 & 0.3270 & 1.0000 & 0.4920 & 0.2610\\
           &     &     & MMP & 0.0119 (0.0262) & 0.9881 & 0.0000 & 0.0000 & 0.0000 & 0.9880\\
    \midrule
           & 0.3 & 0.3 & MEMM & 0.9872 (0.0204) & 0.0128 & 1.0000 & 1.0000 & 1.0000 & 1.0000\\
           &     &     & SIS-MCP & 0.1150 (0.0701) & 0.8850 & 0.9980 & 0.4750 & 0.6380 & 0.8850\\
           &     &     & Path & 0.2400 (0.0063) & 0.7600 & 1.0000 & 1.0000 & 1.0000 & 0.7600\\
           &     &     & TS-Path & 0.1880 (0.0154) & 0.8120 & 1.0000 & 0.7810 & 0.8770 & 0.8120\\
           &     &     & DM1 & 0.7550 (0.0223) & 0.2450 & 0.3180 & 1.0000 & 0.4830 & 0.2450\\
           &     &     & DM2 & 0.3160 (0.0123) & 0.6840 & 0.2390 & 0.3150 & 0.2720 & 0.6840\\
           &     &     & SPCMA & 0.7390 (0.0462) & 0.2610 & 0.3270 & 1.0000 & 0.4930 & 0.2610\\
           &     &     & MMP & 0.0114 (0.0343) & 0.9886 & 0.0000 & 0.0000 & 0.0000 & 0.9890\\
    \bottomrule
  \end{tabular}
\end{table}

\begin{table}[htbp]
  \centering
\scriptsize
  \caption{Simulation results for \textit{partial} mediation (MP=0.5) with $\sigma_Y^2 = 1$,  varying $m$, $q$, $\rho_{\mathbf X}$ and $\rho_{\mathbf M}$. }
  \label{tab:partial_n200}
  \begin{tabular}{cccccccccc}
    \toprule
    $m,q$ & $\rho_{\mathbf X}$ & $\rho_{\mathbf M}$ & Method & MP (SD) & Abs.\ bias & Precision & Recall & F1 & Accuracy\\
    \midrule
m=20, q=20 & 0.0 & 0.0 & MEMM     & 0.3089\,(0.0610) & 0.1911 & 0.5028 & 0.9880 & 0.6626 & 0.7405\\
           &     &     & SIS-MCP & 0.2580\,(0.0238) & 0.2420 & 0.9690 & 0.9950 & 0.9810 & 0.2420\\
           &     &     & Path  & 0.1120\,(0.0153) & 0.3880 & 1.0000 & 0.4490 & 0.6170 & 0.3880\\
           &     &     & TS-Path       & 0.2110\,(0.0287) & 0.2890 & 0.7600 & 0.6330 & 0.6880 & 0.2890\\
           &     &     & DM1      & 0.8330\,(0.0214) & 0.3330 & 0.3020 & 1.0000 & 0.4630 & 0.3330\\
           &     &     & DM2      & 0.3180\,(0.0401) & 0.1820 & 0.2490 & 0.3170 & 0.2790 & 0.1820\\
           &     &     & SPCMA    & 0.8620\,(0.0763) & 0.3620 & 0.2920 & 1.0000 & 0.4520 & 0.3620\\
           &     &     & MMP      & 0.3080\,(0.0639) & 0.1920 & 0.4140 & 0.5140 & 0.4570 & 0.1920\\
    \midrule
           & 0.3 & 0.0 & MEMM     & 0.4427\,(0.0522) & 0.0573 & 0.7939 & 1.0000 & 0.8791 & 0.9260\\
           &     &     & SIS-MCP & 0.3400\,(0.0959) & 0.1600 & 0.7440 & 0.9460 & 0.8160 & 0.1600\\
           &     &     & Path  & 0.2500\,(0.0000) & 0.2500 & 1.0000 & 1.0000 & 1.0000 & 0.2500\\
           &     &     & TS-Path   & 0.2800\,(0.0273) & 0.2200 & 0.9010 & 1.0000 & 0.9460 & 0.2200\\
           &     &     & DM1      & 0.8510\,(0.0233) & 0.3510 & 0.2940 & 1.0000 & 0.4540 & 0.3510\\
           &     &     & DM2      & 0.3130\,(0.0063) & 0.1870 & 0.2470 & 0.3090 & 0.2740 & 0.1870\\
           &     &     & SPCMA    & 0.8610\,(0.0771) & 0.3610 & 0.2930 & 1.0000 & 0.4520 & 0.3610\\
           &     &     & MMP      & 0.2500\,(0.0321) & 0.2500 & 0.9900 & 0.9910 & 0.9910 & 0.2500\\
    \midrule
           & 0.0 & 0.3 & MEMM     & 0.4159\,(0.0856) & 0.0841 & 0.7636 & 1.0000 & 0.8586 & 0.9105\\
           &     &     & SIS-MCP & 0.3130\,(0.0745) & 0.1870 & 0.8070 & 0.9680 & 0.8690 & 0.1870\\
           &     &     & Path  & 0.1110\,(0.0157) & 0.3890 & 1.0000 & 0.4460 & 0.6140 & 0.3890\\
           &     &     & TS-Path       & 0.2130\,(0.0298) & 0.2870 & 0.7540 & 0.6340 & 0.6850 & 0.2870\\
           &     &     & DM1      & 0.8350\,(0.0225) & 0.3350 & 0.3010 & 1.0000 & 0.4620 & 0.3350\\
           &     &     & DM2      & 0.3160\,(0.0466) & 0.1840 & 0.2510 & 0.3170 & 0.2800 & 0.1840\\
           &     &     & SPCMA    & 0.8620\,(0.0755) & 0.3620 & 0.2920 & 1.0000 & 0.4520 & 0.3620\\
           &     &     & MMP      & 0.3200\,(0.0570) & 0.1800 & 0.4180 & 0.5380 & 0.4680 & 0.1800\\
    \midrule
           & 0.3 & 0.3 & MEMM     & 0.4933\,(0.0361) & 0.0067 & 0.9493 & 1.0000 & 0.9721 & 0.9845\\
           &     &     & SIS-MCP & 0.4260\,(0.0800) & 0.0740 & 0.5630 & 0.9310 & 0.6920 & 0.0745\\
           &     &     & Path  & 0.2500\,(0.0001) & 0.2500 & 1.0000 & 1.0000 & 1.0000 & 0.2500\\
           &     &     & TS-Path       & 0.2870\,(0.0274) & 0.2130 & 0.8790 & 1.0000 & 0.9330 & 0.2130\\
           &     &     & DM1      & 0.8420\,(0.0235) & 0.3420 & 0.2970 & 1.0000 & 0.4580 & 0.3420\\
           &     &     & DM2      & 0.3170\,(0.0053) & 0.1830 & 0.2520 & 0.3190 & 0.2810 & 0.1830\\
           &     &     & SPCMA    & 0.8380\,(0.0759) & 0.3380 & 0.3010 & 1.0000 & 0.4620 & 0.3380\\
           &     &     & MMP      & 0.2530\,(0.0385) & 0.2470 & 0.9840 & 0.9950 & 0.9900 & 0.2470\\
    \midrule
m=50, q=50 & 0.0 & 0.0 & MEMM     & 0.3691\,(0.0425) & 0.1309 & 0.6299 & 0.9625 & 0.7579 & 0.8496\\
           &     &     & SIS-MCP & 0.1980\,(0.0204) & 0.3020 & 1.0000 & 0.8260 & 0.9050 & 0.3020\\
           &     &     & Path  & 0.0506\,(0.0134) & 0.4494 & 1.0000 & 0.2110 & 0.3480 & 0.4490\\
           &     &     & TS-Path       & 0.0869\,(0.0233) & 0.4131 & 0.8500 & 0.3070 & 0.4500 & 0.4130\\
           &     &     & DM1      & 0.7370\,(0.0204) & 0.2370 & 0.3270 & 1.0000 & 0.4930 & 0.2370\\
           &     &     & DM2      & 0.3180\,(0.0321) & 0.1820 & 0.2420 & 0.3200 & 0.2750 & 0.1820\\
           &     &     & SPCMA    & 0.7560\,(0.0620) & 0.2560 & 0.3190 & 1.0000 & 0.4830 & 0.2560\\
           &     &     & MMP      & 0.1640\,(0.0570) & 0.3360 & 0.0015 & 0.0011 & 0.0013 & 0.3360\\
    \midrule
           & 0.3 & 0.0 & MEMM     & 0.4591\,(0.0456) & 0.0409 & 0.9962 & 1.0000 & 0.9980 & 0.9990\\
           &     &     & SIS-MCP & 0.1700\,(0.0564) & 0.3300 & 1.0000 & 0.7090 & 0.8260 & 0.3300\\
           &     &     & Path  & 0.2400\,(0.0001) & 0.2600 & 1.0000 & 1.0000 & 1.0000 & 0.2600\\
           &     &     & TS-Path       & 0.1990\,(0.0215) & 0.3010 & 1.0000 & 0.8270 & 0.9050 & 0.3010\\
           &     &     & DM1      & 0.7500\,(0.0218) & 0.2500 & 0.3200 & 1.0000 & 0.4850 & 0.2500\\
           &     &     & DM2      & 0.3170\,(0.0329) & 0.1830 & 0.2410 & 0.3180 & 0.2740 & 0.1830\\
           &     &     & SPCMA    & 0.7330\,(0.0661) & 0.2670 & 0.3300 & 1.0000 & 0.4950 & 0.2330\\
           &     &     & MMP      & 0.0126\,(0.0302) & 0.4874 & 0.0000 & 0.0000 & 0.0000 & 0.4870\\
    \midrule
           & 0.0 & 0.3 & MEMM     & 0.4586\,(0.0447) & 0.0414 & 0.9954 & 1.0000 & 0.9976 & 0.9988\\
           &     &     & SIS-MCP & 0.1950\,(0.0533) & 0.3050 & 1.0000 & 0.8120 & 0.8960 & 0.3050\\
           &     &     & Path  & 0.0502\,(0.0127) & 0.4498 & 1.0000 & 0.2090 & 0.3450 & 0.4500\\
           &     &     & TS-Path       & 0.0869\,(0.0285) & 0.4131 & 0.8650 & 0.3120 & 0.4580 & 0.4130\\
           &     &     & DM1      & 0.7310\,(0.0216) & 0.2310 & 0.3300 & 1.0000 & 0.4960 & 0.2310\\
           &     &     & DM2      & 0.3170\,(0.0419) & 0.1830 & 0.2380 & 0.3140 & 0.2710 & 0.1830\\
           &     &     & SPCMA    & 0.7490\,(0.0642) & 0.2510 & 0.3230 & 1.0000 & 0.4870 & 0.2490\\
           &     &     & MMP      & 0.1600\,(0.0494) & 0.3400 & 0.0005 & 0.0003 & 0.0004 & 0.3400\\
    \midrule
           & 0.3 & 0.3 & MEMM     & 0.4950\,(0.0341) & 0.0050 & 1.0000 & 1.0000 & 1.0000 & 1.0000\\
           &     &     & SIS-MCP & 0.1350\,(0.0463) & 0.3650 & 1.0000 & 0.5640 & 0.7170 & 0.3650\\
           &     &     & Path  & 0.2400\,(0.0001) & 0.2600 & 1.0000 & 1.0000 & 1.0000 & 0.2600\\
           &     &     & TS-Path      & 0.1950\,(0.0209) & 0.3050 & 1.0000 & 0.8130 & 0.8970 & 0.3050\\
           &     &     & DM1      & 0.7590\,(0.0194) & 0.2410 & 0.3170 & 1.0000 & 0.4810 & 0.2590\\
           &     &     & DM2      & 0.3170\,(0.0358) & 0.1830 & 0.2400 & 0.3170 & 0.2730 & 0.1830\\
           &     &     & SPCMA    & 0.7430\,(0.0497) & 0.2430 & 0.3250 & 1.0000 & 0.4900 & 0.2430\\
           &     &     & MMP      & 0.0128\,(0.0420) & 0.4872 & 0.0000 & 0.0000 & 0.0000 & 0.4870\\
    \bottomrule
  \end{tabular}
\end{table}

Across all simulation scenarios, the proposed MEMM method consistently achieves the lowest absolute bias in mediation proportion (MP) estimation and exhibits the best or near-best weight selection performance (Table~\ref{tab:complete_n200},~\ref{tab:partial_n200}). Under complete mediation, MEMM accurately recovers the true MP with negligible bias, especially when both exposure and mediator correlations are high ($\rho_X, \rho_M = 0.3$), suggesting that the aggregation structure is particularly effective in the presence of correlated features. The advantage of MEMM remained even with larger number of exposures and mediators ($m=50$,$q=50$). In contrast, regularization based methods (SIS-MCP, Path and TS-Path) accurately selected the weights but persistently displayed large bias in MP estimation with even worse performance when the correlation and number of features increase. The dimension reduction methods (DM1/DM2 and SPCMA), on the other hand, had smaller bias for MP estimation but poorer performance in weight selection (e.g. selecting many false positive weights thus low precision), making it harder to interpret the final aggregators. The hybrid method MMP, originally designed for one exposure only, performed poorly in this setting especially when the number of features and between-feature correlation increased. Our method MEMM well balanced the accurate estimation of MP that captures the predominant mediation pathway and accurate selection of weights for the best interpretability of the aggregated exposures and mediators, well suited to the multi-exposure-to-multi-mediator problem. 

Our advantage over other competing methods remained with partial mediation (MP=0.5). Specifically, MEMM maintained a low absolute bias in MP estimation for partial mediation without over-estimating MP, showing the strength of our objective function that balanced the likelihood, mediation proportion and weight estimation. By contrast, DM1/DM2 and SPCMA tend to over-estimate the mediation proportion. The weight selection performance also dropped dramatically for the regularization based methods with partial mediation. As the noise level increased (\(\sigma_Y^2 = 3\); Table S1), all methods had worsened performance but MEMM remained as the best performer especially when the between-feature correlation increased. 

When $a$ was misspecified, the MP estimation and the overall selection of true zero and nonzero weights remained accurate (Table S2). Misspecification of $b$ had a stronger negative impact on the MP estimation, though the precision of weight selection by our method remained high, and the overall accuracy of selection increased as the correlation increased (Table S2). When the true MP was smaller (MP=0.25), or the signal strength was weakened  ($c=0.1$), our method performed consistently well with close to true MP estimation and high accuracy in selecting the true zero and nonzero weights (Table S3-S4). Table S5 examines the impact of scaling up the LASSO penalty parameters $\lambda_a$ and $\lambda_b$ on the performance of our method under complete mediation (MP=1) when n=200, m=q=20 and ($\rho_X,\rho_M)=(0.3,0.3)$.
The new set of penalty parameters $\lambda_a=C_\lambda \lambda_a^{\text{base}}$, $\lambda_b= C_\lambda \lambda_b^{\text{base}}$, where $\lambda_a^{\text{base}}=\lambda_b^{\text{base}}=0.15$ falls in the original range of grid search, $C_\lambda \in \{1,2,5,10,15\}$ is the scaling factor, $\lambda_n$ is fixed to be 0.1. The MP estimation and the selection of true zero and nonzero weights remained highly accurate with modest penalization ($C_\lambda=1,2,5$, which is of order $\sqrt{\log n}$; Table S5). With too strong penalization ($C_\lambda=10,15$), the model is subject to over-shrinkage so the recall and F1 score dropped rapidly.

%$\mathrm{F}_1$ scores exceeding 0.90 in most complete mediation settings, and remains above 0.75 in more challenging partial mediation regimes (e.g., mediation proportion = 0.5 with $m=q=50$). Notably, MEMM exhibits near-perfect recall across settings, indicating its ability to recover truly active mediators with minimal false positives. By contrast, SIS-MCP suffers from low precision due to the inclusion of many irrelevant variables, while Pathway Lasso tends to miss relevant mediators and exposures, resulting in low recall. DM1/DM2 and SPCMA demonstrate moderate performance in low-noise scenarios but their precision deteriorates under correlated mediators, highlighting their sensitivity to multicollinearity. Interestingly, MEMM remains robust to moderate correlation in either exposures or mediators, showing minimal degradation in either MP estimation or selection accuracy.

%When the problem dimension increases from $(m,q)=(20,20)$ to $(50,50)$, all methods experience higher variability, yet MEMM retains the greatest stability. Under amplified noise levels (\(\sigma_2^2 = 3\)), MEMM’s absolute bias increases slightly but still remains one-half to one-third that of the next best-performing method, emphasizing the robustness conferred by its penalized optimization framework.

%Overall, these results support the effectiveness of MEMM in both mediation effect quantification and mediator selection, particularly under high-dimensional, correlated scenarios where existing methods either overfit or fail to identify the true mediation structure.

\section{Application to Real Data}
{Nicotine dependence is a multifaceted brain disorder with a significant heritable component, and a leading modifiable risk factor for cancer, pulmonary, cardiovascular and cerebrovascular disorders \citep{bartal2001health}. Large-scale GWAS and TWAS on different measures of tobacco smoking behaviors \citep{erzurumluoglu2020meta,xu2020genome,buchwald2021genome,ye2021meta} have reported multiple replicable genetic loci and clusters of genes potentially associated with nicotine dependence. In the same time, neuroimaging studies using magnetic resonance imaging (MRI) data have revealed alterations in both gray and white matter structural and functional characteristics associated with nicotine addiction \citep{kuhn2012brain,gray2020associations}, yet the pathways from genes to brain alternations to nicotine addiction remained unclear. In this section, we applied our MEMM method that integrates genetically regulated expression (GReX), resting-state functional MRI (rs-fMRI), and smoking behavior data from UK Biobank (UKB) to evaluate the potential ``gene-brain functional connectivity-nicotine addiction'' mediation pathway, exploring how genetic predispositions influence brain circuits, which in turn make people more addicted to nicotine.}

As in most other studies, we treated cigarettes per day (CPD) as the main outcome of interest to assess nicotine dependence, and also considered the former and current smoking status in selecting the samples for analysis following the validated diagnostic criteria for nicotine dependence \citep{Heatherton1991}. We first obtained the GReX data from the genotype data in UKB and the GTEx \citep{gtex2020gtex} reference panel data in three addiction-relevant brain tissues—Nucleus Accumbens (NAcc), Caudate (Ca), and Putamen (Pu), all subregions of the basal ganglia. Following our previous works of TWAS on cigarettes per day (CPD) using MultiXcan \citep{barbeira2019integrating,ye2021meta,wang2024tips}, we selected a group of $m=21$ genes most significantly associated with CPD on chromosome 15 ($p<1e-8$; see Fig S1 for the TWAS Manhattan plot) and used their GReX values aggregated over three brain tissues as exposures in our study. Region of Interest (ROI)-to-ROI functional connectivity (FC) features were computed from preprocessed rs-fMRI data. We restricted the analysis to basal ganglia and subcortical lobe regions most relevant to nicotine addiction and kept $q=635$ ROI-to-ROI FC measures as potential mediators for our analysis (Fig S2). Our final analytical sample includes $n=710$ participants with complete exposure, mediator and outcome data. As we can see, many genes and many FC measures in these samples are highly correlated with each other (Fig 2A and 2B), laying out the basis for our reduced rank based modeling of both exposures and mediators. 

% Based on biological relevance—including known loci such as chromosome 15q25—and signal strength, we selected a cluster of nine genes for downstream analysis in Figure \ref{fig:2}(A). These comprised the final exposure matrix with $n = 710$ subjects and $m = 21$ genes. The mediator matrix $\mathbf{M} \in \mathbb{R}^{710 \times 635}$ represented the full set of FC features across cortical and subcortical regions, which exhibited internal correlation patterns (e.g., subcortical clusters) visualized via heatmaps. The outcome variable was CPD, defined as the average number of cigarettes smoked per day. Participants with implausible or missing CPD values were excluded from the analysis.

%The penalized mediation model was fitted using the objective function introduced in Section 2, which integrates squared loss from three linear models—outcome on exposure, mediator on exposure, and outcome on both exposure and mediator—along with a sparsity-inducing $\ell_1$ penalty and a mediation-enhancing coupling term. 

%We performed the ADMM-based optimization with block coordinate descent. 

%We fitted our model to this data and solved using ADMM algorithm. The tuning parameters $(\lambda_a, \lambda_b)$ were selected via five-fold cross-validation. The coupling parameter $\lambda_n$, which controls the strength of the mediated effect proportion, was chosen from a small grid of candidate values based on prediction performance and numerical stability. 

We applied a five-fold cross-validation over a coarse-to-fine grid to select the best $\lambda_a$, $\lambda_b$ and $\lambda_n$, where $\lambda_a, \lambda_b \in \{ 0.02,0.05,0.10,0.20,0.30,0.50 \}$, $\lambda_n \in \{0.02,0.05,0.08,0.10,0.15\}$. 
The final set of penalty parameters selected was $(\lambda_a,\lambda_b,\lambda_n)=(0.30,\,0.10,\,0.08)$. The ADMM algorithm converged in $9$ iterations, and the full run (including cross-validation) completed in about 20 minutes. The estimated MP is approximately 0.40, indicating that roughly half of the genetic effects on CPD on chromosome 15 was mediated through the brain FC between ROIs retained in the model. This result suggested that the functional connectome between specific brain regions, the brain’s dynamic wiring diagram, may serve as a key conduit translating genetic influences into brain disorders like nicotine addiction. At the same time, the non-mediated portion of the effect suggested alternative biological pathways (e.g., hormonal, metabolic, or peripheral nervous system routes). For the exposures, we identified a coherent subset of six genes with nonzero weights (Table 3). All the three genes in the \textit{CHRNA5-CHRNA3-CHRNB4} nicotinic receptor subunit gene cluster \citep{lassi2016chrna5} were selected, affirming their role in nicotine response and smoking behavior. These genes encoded subunits of nicotinic acetylcholine receptors (nAChRs), and their altered expression is known to influence nicotine intake and receptor sensitivity \citep{fowler2011alpha5}. \textit{PSMA4} and \textit{CTSH}, both encoded for protease involved in protein degradation pathways, may influence brain signaling pathways, neuroadaptations and neuroinflammation linked to nicotine addiction \citep{rezvani2007nicotine,sherva2016genetic}. 

%Partial mediation of this magnitude was consistent with prior findings in neurogenetic mediation \citep{huckins2020gene} and supports the interpretation that both direct and indirect mechanisms contribute to smoking behavior.

%\textit{CTSH} encodes cathepsin H, a lysosomal cysteine protease involved in neuroinflammation and neuronal remodeling, processes implicated in substance dependence and neuroplasticity in response to chronic nicotine exposure \citep{li2020ctsh}.
%\textit{GRK4} encodes G protein-coupled receptor kinase 4, which modulates dopaminergic signaling and has been associated with blood pressure regulation and reward processing—both relevant to the reinforcing effects of nicotine \citep{zhang2017grk4}.
%\textit{CRABP1} (cellular retinoic acid-binding protein 1) regulates retinoic acid signaling pathways involved in neural development and has been linked to stress responses and synaptic plasticity, potentially influencing susceptibility to nicotine dependence \citep{tsai2011crabp1}. 

For the mediators, we identified a sparse subset of 32 FC measures spanning several addiction-relevant regions (Table S6). Among these were functional connectivity between middle frontal gyrus (MFG), precentral gyrus (PrG), thalamus and superior parietal lobule (SPL) regions. MFG has played a significant role in inhibitory control and reward processing of nicotine \citep{hong2009association}. The functional connectivity between SPL and the other brain regions (e.g. the insula for reward processing and craving) has been found correlated with the severity of nicotine dependence \citep{ghahremani2023nicotine}, and the thalamus is enriched in nicotinic receptors \citep{spurden1997nicotinic}. Our findings suggest that genetic variation affecting expression of nicotinic receptors may exert influence on CPD via alterations in MFG, SPL and thalamus subregions.

\begin{figure*}[htbp] % The star (*) spans the figure across both columns
  \centering
    \includegraphics[scale=0.6]{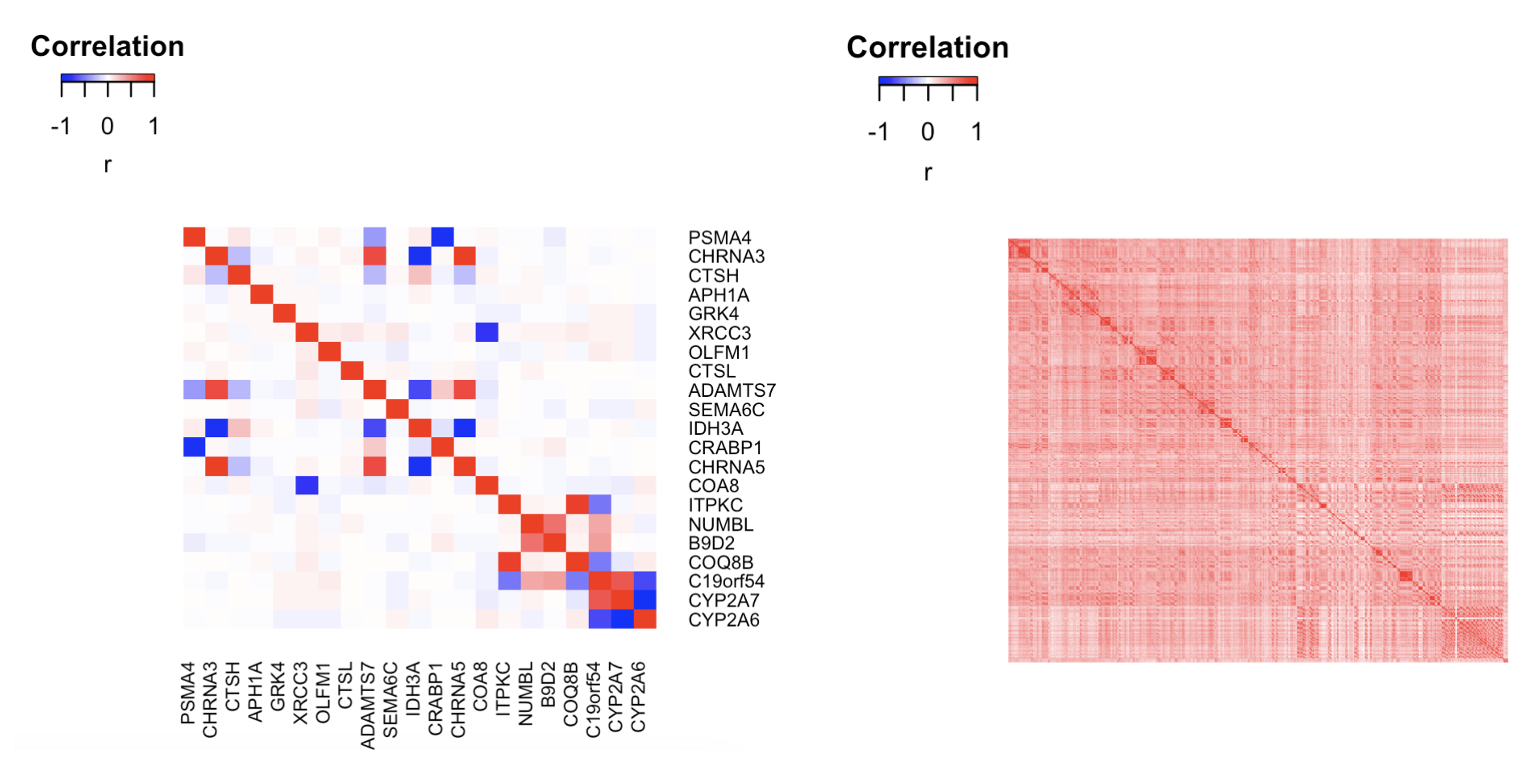} % Make sure to use the correct file path
\caption{
(A) Pairwise correlations among the expression of 21 genes used as exposures in our analysis.  
(B) Pairwise correlations among the 635 ROI-to-ROI functional connectivity (FC) measures used as mediators in our analysis.  
}
  \label{fig:2}
\end{figure*}

\begin{table}[htbp]
  \caption{List of genes selected by MEMM for the ``gene-brain FC-nicotine addiction'' mediation pathway}
  \small
  \label{tab:gene}
  \begin{tabular}{p{3cm} p{13cm}}
    \toprule 
      Gene symbols & Gene function \\
     \midrule
        \textit{CHRNA3} & encodes a subunit of the neuronal nicotinic acetylcholine receptor (nAChR) called alpha-3 subunit. \\
        \textit{CHRNA5} & encodes a subunit of the neuronal nicotinic acetylcholine receptor (nAChR) called alpha-5 subunit. \\
        \textit{PSMA4} & encodes a subunit of the 20S proteasome, part of the ubiquitin-proteasome system responsible for protein degradation within cells. \\
        \textit{GRK4} & encodes G protein-coupled receptor kinase 4, which plays a key role in regulating blood pressure. \\
        \textit{CRABP1} & encodes a carrier protein binding to retinoic acid and facilitating its transport within the cell.  \\
        \textit{CTSH} & encodes cathepsin H, a lysosomal cysteine protease crucial for protein degradation within cells. \\
    \bottomrule
  \end{tabular}
\end{table}

\section{Discussion}
In this paper, we proposed a new multi-exposure-to-multi-mediator mediation model for imaging genetic study of brain disorders to understand how the multivariate genetic risk factors impact the brain disorders via changing the multivariate neuroimaging features, including both structural and functional changes. We constructed sparse aggregated variables that captured much of the indirect effect in the mediation pathway, balancing the estimation of mediation proportion and feature selection for interpretability in a mediation model with multiple exposures and mediators. We proposed a fast algorithm based on ADMM to solve the optimization and theoretically justified the convergence guarantees of the algorithm and consistency of the estimator. An efficient R package called ``MEMM'' was developed (https://github.com/nwang123/MEMM) to implement the method.

%Unlike two-stage or screening-based procedures, the proposed method integrates selection and estimation in a unified penalized objective, reducing the risk of overlooking mediators with weak marginal signals but strong joint effects.

Our model is among the very few that perform mediation analysis with multiple exposures and multiple mediators. In imaging genetic study specifically, it is widely known that structural or functional changes in the brain can partially explain the genetic effect on brain disorders. However, with the large number of relevant genes and neuroimaging features, identifying the predominant mediation pathway through which primary genetic determinants of brain disorders exert their effects via functional or structural alterations in specific brain regions remains challenging. Our method proposed one solution that identified the best mediation pathway with large mediation proportion while in the same time selecting the most critical genes and regional neuroimaging features involved in the pathway for better interpretability. In the examples we showed, we assumed the number of both exposures and mediators are of moderate size (at most $n$) as both need to be related to the outcome of interest. When both exposures and mediators are high-dimensional ($\gg n$), the algorithm can be less stable and computationally intensive. In this case, we encourage a two-stage approach by first screening out noisy exposures and mediators simultaneously using advanced high-dimensional-to-high-dimensional sure screening techniques \citep{ke2022high} and then performing our method. 

%The method’s performance is supported by comprehensive simulations under both complete and partial mediation scenarios, demonstrating robustness to correlation structures and noise levels. In particular, our approach consistently achieves lower absolute bias in mediation proportion estimates and higher recovery of active variables compared to established alternatives. These properties are crucial in scientific contexts where interpretability, reproducibility, and statistical efficiency are paramount.

%Application to a real-world neurogenetics dataset further underscores the value of the method. Our findings corroborate known associations involving the CHRNA5–CHRNA3–CHRNB4 locus and highlight potential brain regions mediating these genetic effects on smoking behavior. The ability to identify such interpretable, biologically grounded pathways exemplifies the method’s utility in advancing scientific understanding in genomics and neuroscience.

The current method is built on the basis of a linear mediation model. For genetic and neuroimaging features, it is common that more complex nonlinear relationship and mediation effect may exist. Future work could focus on incorporating nonlinearity into our current mediation model \citep{loh2022nonlinear} and adopting nonlinear dimension reduction in the reduce rank aggregation \citep{lee2007nonlinear}. Lastly, establishing the true causal mediation relies on major identification assumptions \citep{imai2010general} that are usually not directly testable for the large-scale genetic and imaging data in cross-sectional studies like UKB. Future research could focus on longitudinal cohorts such as Adolescent Brain Cognitive Development (ABCD) study \citep{casey2018adolescent} that includes longitudinal imaging mediators and outcome with possible temporal ordering to validate the identified mediation pathways.

%Future work may extend this framework to incorporate nonlinear mediation effects, interactions, or longitudinal outcomes, and to explore its connections with structural equation modeling in latent-variable settings. Additionally, theoretical investigations into the identifiability and asymptotic properties of the proposed estimators may further strengthen the methodological foundation. Overall, this study provides a rigorous and practical contribution to the growing literature on high-dimensional causal mediation, with broad applicability across biomedical, social, and computational sciences.

\section*{Author contributions}

\textbf{Neng Wang:} methodology, theoretical proof, simulation, data collection and analysis, software, original manuscript preparation and editing. \textbf{Eric V Slud:} supervision of the project, theoretical proof, manuscript review and editing. \textbf{Tianzhou Ma:} conceptualization, supervision of the project, manuscript review and editing.

\section*{Acknowledgments}
This work was supported in part by the National Institutes of Health (NIDA 1K01DA059603) and Grand Challenge Grant by the University of Maryland.

\section*{Appendix A.}

\begin{proof}[Proof of Theorem~\ref{thm:conv_rate}]
Define the augmented Lagrangian
\[
\mathcal{L}_\rho\bigl(\mathbf{a},\mathbf{b},\mathbf{z}_a,\mathbf{z}_b,\mathbf{u}_a,\mathbf{u}_b\bigr)
:=F(\mathbf{a},\mathbf{b})
+G_a(\mathbf{z}_a)
+G_b(\mathbf{z}_b)
+\frac{\rho}{2}\|\mathbf{a}-\mathbf{z}_a+\mathbf{u}_a\|_2^2
+\frac{\rho}{2}\|\mathbf{b}-\mathbf{z}_b+\mathbf{u}_b\|_2^2,
\]
where $G_a(\mathbf{z}_a)=\lambda_a \|\mathbf{z}_a\|_1$ and $G_b(\mathbf{z}_b)=\lambda_b \|\mathbf{z}_b\|_1$ and abbreviate
\[
\mathcal{L}_\rho^k
:=\mathcal{L}_\rho\bigl(\mathbf{a}^k,\mathbf{b}^k,\mathbf{z}_a^k,\mathbf{z}_b^k,\mathbf{u}_a^k,\mathbf{u}_b^k\bigr),
\quad
\Delta_k
:=\|\mathbf{a}^{k+1}-\mathbf{z}_a^{k+1}\|_2^2
     +\|\mathbf{b}^{k+1}-\mathbf{z}_b^{k+1}\|_2^2.
\]

\paragraph{Step 1: Descent under the $\mathbf{a}$– and $\mathbf{b}$–updates.}
By Assumption~\ref{asmp:Assumption} (convexity and unique minimisers of the block subproblems), the exact minimisation in
each of the $\mathbf{a}$ and $\mathbf{b}$ blocks yields
\[
\mathcal{L}_\rho\bigl(\mathbf{a}^{k+1},\mathbf{b}^{k+1},\mathbf{z}_a^k,\mathbf{z}_b^k,\mathbf{u}_a^k,\mathbf{u}_b^k\bigr)
\;\le\;
\mathcal{L}_\rho^k.
\tag{1}
\]

\paragraph{Step 2: Descent under the $\mathbf{z}$–updates.}
The proximal updates
\(
\mathbf{z}_a^{k+1}
= S_{\lambda_a/\rho}\bigl(\mathbf{a}^{k+1}+\mathbf{u}_a^k\bigr)
\)
and
\(
\mathbf{z}_b^{k+1}
= S_{\lambda_b/\rho}\bigl(\mathbf{b}^{k+1}+\mathbf{u}_b^k\bigr)
\)
satisfy
\[
\mathcal{L}_\rho\bigl(\mathbf{a}^{k+1},\mathbf{b}^{k+1},
                      \mathbf{z}_a^{k+1},\mathbf{z}_b^{k+1},
                      \mathbf{u}_a^k,\mathbf{u}_b^k\bigr)
\;\le\;
\mathcal{L}_\rho\bigl(\mathbf{a}^{k+1},\mathbf{b}^{k+1},
                      \mathbf{z}_a^k,\mathbf{z}_b^k,
                      \mathbf{u}_a^k,\mathbf{u}_b^k\bigr)
-\frac{\rho}{2}\,\Delta_k.
\tag{2}
\]

\paragraph{Step 3: Summability of the primal residuals.}
The dual updates
\(
\mathbf{u}_a^{k+1}=\mathbf{u}_a^k+\mathbf{a}^{k+1}-\mathbf{z}_a^{k+1},
\quad
\mathbf{u}_b^{k+1}=\mathbf{u}_b^k+\mathbf{b}^{k+1}-\mathbf{z}_b^{k+1}
\)
leave \(\mathcal{L}_\rho\) unchanged.  Combining (1) and (2) gives
\[
\mathcal{L}_\rho^{\,k+1}
\;\le\;
\mathcal{L}_\rho^{\,k}
-\frac{\rho}{2}\,\Delta_k.
\]
Summing from \(k=0\) to \(K-1\) yields
\[
\mathcal{L}_\rho^{\,K}
+\frac{\rho}{2}\sum_{k=0}^{K-1}\Delta_k
\;\le\;
\mathcal{L}_\rho^{\,0}.
\]
By Assumption~\ref{asmp:Assumption} and the fact that the domain of the iterates is restricted to the admissible region \(\mathcal C_{r_0,\delta}\), the augmented Lagrangian is bounded below; hence
\(\sum_{k=0}^\infty \Delta_k<\infty\)
and, in particular, \(\Delta_k\to0\).

\paragraph{Step 4: Boundedness and stationarity.}
Since \((\mathbf{a}^k,\mathbf{b}^k)\in\mathcal{C}\) for all \(k\) and
\(\mathcal{C}\) is compact, the primal iterates are bounded.  Dual
iterates remain bounded because their increments
\(\mathbf{a}^{k+1}-\mathbf{z}_a^{k+1}\) and
\(\mathbf{b}^{k+1}-\mathbf{z}_b^{k+1}\) vanish.  Hence cluster points
exist.

Exact optimality in the \(\mathbf{a}\) block gives
\[
0
\;=\;
\nabla_{\!\mathbf{a}}F(\mathbf{a}^{k+1},\mathbf{b}^k)
+\rho\bigl(\mathbf{a}^{k+1}-\mathbf{z}_a^k+\mathbf{u}_a^k\bigr),
\]
and after the dual update one shows
\(\nabla_{\!\mathbf{a}}\Phi(\mathbf{a}^{k+1},\mathbf{b}^{k+1})
+\rho(\mathbf{a}^{k+1}-\mathbf{z}_a^{k+1})=0\).
Similarly for \(\mathbf{b}\).  Taking norms and using
\(\Delta_k\to0\) yields
\(\|\nabla\Phi(\mathbf{a}^k,\mathbf{b}^k)\|\to0\).  Thus every cluster
point is a first‐order stationary point of \(\Phi\) on \(\mathcal{C}\),
proving part~(i).

\paragraph{Step 5: Convergence rates via the KŁ inequality.}
From Steps 1–3 we have a descent
\(\Phi(\mathbf{a}^k,\mathbf{b}^k)-\Phi(\mathbf{a}^{k+1},\mathbf{b}^{k+1})
\ge\tfrac{\rho}{2}\,\Delta_k\)
and \(\sum\Delta_k<\infty\).  By the Kurdyka–Łojasiewicz (KŁ) property assumed in Assumption \ref{asmp:KL}, there exist constants \(c>0\) and \(\theta\in[0,1)\) such that,
in a neighbourhood of any cluster point,
\[
\|\nabla\Phi(\mathbf{a},\mathbf{b})\|
\;\ge\;
c\,\bigl[\Phi(\mathbf{a},\mathbf{b})-\Phi^\star\bigr]^{\theta}.
\]
Standard arguments (cf.\ \citep[Thm.~2]{hong2016convergence}) then yield
the global \(\mathcal{O}(1/K)\) rate in part~(ii).  If additionally \(F\) is strongly convex in
one block (with the other block fixed), the KŁ exponent satisfies \(\theta=\tfrac12\),
giving the local linear rate in part~(iii).
\end{proof}

\bibliographystyle{elsarticle-num}
\bibliography{references}

\end{document}